\documentclass{amsart}

\usepackage[foot]{amsaddr}
\usepackage{amsfonts, amstext, amsmath, amssymb, stmaryrd, mathrsfs, calc, amsthm, xspace, mathtools}  
\usepackage{bm}         
\usepackage{a4wide}     
\usepackage{enumerate}  
\usepackage{macros}     

\usepackage{graphicx}   

\title[Reachability in Two-Dimensional Vector Addition Systems with
  States]{Reachability in Two-Dimensional Vector Addition Systems with
  States is PSPACE-complete}

\author[M. Blondin]{Michael Blondin$^{1,2,\ast}$}
\address{$^1$DIRO, Universit\'{e} de Montr\'{e}al, Canada}
\email{blondimi@iro.umontreal.ca}
\thanks{$^\ast$Supported by the Fonds qu\'{e}b\'{e}cois
  de la recherche sur la nature et les technologies and by the French
  Centre nationale de la recherche scientifique.}

\author[A. Finkel]{Alain Finkel$^{2,\dagger}$}
\address{$^2$Laboratoire Sp\'ecification et V\'erification (LSV),
  ENS de Cachan \& CNRS, France}
\email{finkel@lsv.ens-cachan.fr}
\thanks{$^\dagger$Supported by the French Agence nationale
  de la recherche, \textsc{ReacHard} (grant ANR-11-BS02-001).}

\author[S. G\"oller]{Stefan G\"{o}ller$^{2,\ddagger,\mathsection}$}
\email{goeller@lsv.ens-cachan.fr}
\thanks{$^\ddagger$Parts of this work were carried out while the
  author was at Technische Universit\"{a}t M\"{u}nchen, Germany.}

\author[C. Haase]{Christoph Haase$^{2,\dagger,\mathsection}$}
\email{haase@lsv.ens-cachan.fr}
\thanks{$^\mathsection$Supported by Labex Digicosme, Univ. Paris-Saclay, 
  project VERICONISS}

\author[P. McKenzie]{Pierre McKenzie$^{1,2,\mathparagraph}$}
\email{mckenzie@iro.umontreal.ca}
\thanks{$^\mathparagraph$Supported by the Natural Sciences and
  Engineering Research Council of Canada and by the ``Chaire Digiteo,
  ENS Cachan - {\'E}cole Polytechnique''.}

\begin{document}

\begin{abstract}
  Determining the complexity of the reachability problem for vector
  addition systems with states (VASS) is a long-standing open problem
  in computer science. Long known to be decidable, the problem to this
  day lacks any complexity upper bound whatsoever. In this paper,
  reachability for two-dimensional VASS is
  shown \PSPACE-complete. This improves on a previously known doubly
  exponential time bound established by Howell, Rosier, Huynh and Yen
  in 1986. The coverability and boundedness problems are also noted to
  be \PSPACE-complete. In addition, some complexity results are given
  for the reachability problem in two-dimensional VASS and in integer
  VASS when numbers are encoded in unary.
\end{abstract}

\maketitle 

\section{Introduction}

{\em Petri nets} have a long history. Since their introduction
\cite{pe62} by Carl Adam Petri in 1962, thousands of papers on Petri
nets have been published. Nowadays, Petri nets find a variety of
applications, ranging, for instance, from modeling of biological,
chemical and business processes to the formal verification of
concurrent programs, see
\emph{e.g.}~\cite{HGD08,RLM96,vdA98,GS92,BCR01}. For the analysis of
algorithmic properties of Petri nets, in the contemporary literature
they are often equivalently viewed as \emph{vector addition systems
  with states (VASS)}, and we will adopt this view in the remainder of
this paper. A VASS comprises a finite-state controller with a finite
number of counters ranging over the natural numbers. The number of
counters is usually referred to as the dimension of the VASS, and we
write $d$-VASS when we talk about VASS in dimension $d$. When taking a
transition, a VASS can add or subtract an integer from a counter,
provided that the resulting counter values are greater than or equal
to zero; otherwise the transition is blocked. A configuration of a
VASS is a tuple consisting of a control state and an assignment of the
counters to natural numbers. The central decision problem for VASS is
\emph{reachability}: given two configurations, is there a path
connecting them in the infinite graph induced by the VASS?

Even clarifying the decidability status of the reachability problem
required tremendous efforts, and it actually took until 1981 for it to
be shown decidable. This was achieved by Mayr~\cite{ma81}, who built
upon an earlier partial proof by Sacerdote and
Tenney~\cite{sate77}. Mayr's argument was then polished and simplified
by Kosaraju~\cite{ko82} in 1982, and Kosaraju's argument was in turn
simplified ten years later by Lambert~\cite{la92}. Only recently
beginning in 2009, Leroux developed, in a series of papers, a
fundamentally different approach to the decidability of the
reachability problem~\cite{Ler09,Ler11,Ler12}. But to this day, no
explicit upper bound on the complexity of the general reachability
problem for VASS is known. A primitive recursive upper bound claim
made in 1998~\cite{bo98} was dismissed in~\cite{ja08}.

Milestones in the work on the computational complexity of the
reachability problem for VASS include Lipton's proof
of \EXPSPACE-hardness~\cite{Lipt76}. This lower bound is independent
of the encoding of numbers, it does however require an unbounded
number of counters. Deciding reachability of VASS in dimension one
assuming unary encoding of numbers is easily seen to be \NL-complete:
the lower bound is inherited from graph reachability and the upper
bound follows from a simple pumping argument.
When numbers are encoded in binary, reachability in VASS in dimension
one is known to be \NP-complete~\cite{HKOW09}. A substantial
contribution towards showing the decidability of the general
reachability problem was made by Hopcroft and Pansiot in 1979, who
showed that reachability in VASS in dimension two is
decidable~\cite{HP79}. To this end, they developed an intricate
algorithm that implicitly exploits the fact that the reachability set
of a VASS in dimension two is semi-linear. Moreover, they could show
that their method breaks down for VASS in any greater dimension, as
the authors exhibited a VASS in dimension three with a reachability
set that is not semi-linear. Yet, aspects of computational complexity
were completely left unanswered in~\cite{HP79}. In 1986, Howell,
Rosier, Huynh and Yen~\cite{DBLP:journals/tcs/HowellRHY86} analyzed
Hopcroft and Pansiot's algorithm and showed that it runs in
nondeterministic doubly-exponential time, independently of whether
numbers are presented in unary or binary. They could improve this
nondeterministic doubly-exponential time upper bound to a
deterministic doubly-exponential one and also identify a family of
VASS in dimension two on which Hopcroft and Pansiot's algorithm
requires doubly-exponential time.
In summary, since 1986 it has been state-of-the-art that reachability
in VASS in dimension two is in 2-\etime, and \NL-hard and \NP-hard,
depending on whether numbers are encoded in unary or binary. Apart
from \EXPSPACE-hardness and decidability, no complexity-theoretic 
upper bound is known for the
complexity of reachability in VASS in any dimension greater than two.

\begin{figure}[t]
  \includegraphics{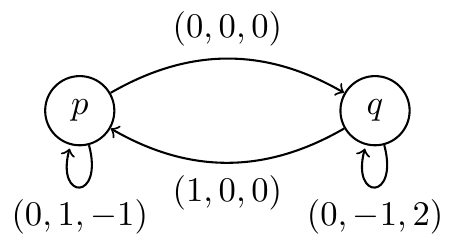}
    
    
  \caption{\label{fig:ex:3vass} Example from~\cite{HP79} of a 3-VASS
    whose reachability set starting in configuration $p(0,0,1)$ is not
    semi-linear.}
\end{figure}

The main contribution of this paper is to show that reachability in
VASS in dimension two is \PSPACE-complete when numbers are encoded in
binary. The \PSPACE\ lower bound follows as an easy consequence of a
recent result by Fearnley and Jurdzi{\'n}ski who
showed \PSPACE-completeness of reachability in bounded one-counter
automata~\cite{FJ13}. Our \PSPACE\ upper bound is obtained from
showing that the length of a run witnessing reachability can be
exponentially bounded in the size of the input, and consequently the
existence of such a run can be decided by a \PSPACE-algorithm. The
difficult and main part of this paper is, of course, to establish the
exponential upper bound on the length of witnessing runs. Our starting
point is a careful analysis of an argument developed by Leroux and
Sutre in~\cite{LS04} for the purpose of showing that reachability
relations of VASS in dimension two can be captured by bounded regular
languages, i.e., speaking in the terminology of~\cite{LS04}, 2-VASS
can be {\em flattened}. More precisely, this means that for any 2-VASS
there is a finite set $S$ of regular languages over the set of
transitions, viewed as an alphabet, each of the form
$u_0v_1^*u_1\cdots v_k^*u_k$ such that for any two reachable
configurations there exists a witnessing run in the language defined
by $S$. The paper of Leroux and Sutre reports that from any VASS in
dimension two it is possible to construct such a bounded language.
This immediately implies that the reachability relation of 2-VASS is
semi-linear. In dimension three, the reachability relation is no
longer semi-linear and hence such bounded languages cannot exist; the
classical example by Hopcroft and Pansiot of a 3-VASS that does not
posses a semi-linear reachability set is depicted in
Fig.~\ref{fig:ex:3vass}. The paper~\cite{LS04} has not appeared as a
fully refereed publication and some proof details are omitted in
it. Thus, while we follow closely the proof strategy presented
in~\cite{LS04}, we provide a complete proof that 2-VASS can be
flattened, and in doing so develop new arguments in order to allow for
a tight analysis of our constructions with the overall goal of
establishing the \PSPACE\ upper bound. In summary, we make the
following contributions:
\begin{itemize}
\item we show \PSPACE-completeness of reachability in VASS in
  dimension two,
\item for showing the former we provide a complete and rigorous proof
  that VASS in dimension two can be flattened by bounded languages
  that have small presentations, and
\item we remark that reachability in VASS in dimension two with numbers
  encoded in unary is \NL-hard and in \NP.
\end{itemize}

The structure of this paper is as follows. In Section~\ref{prelims},
we introduce our notation, give relevant definitions and formally
define vector addition systems with states.  Section~\ref{sec:results}
gives an overview of our main results. In
Section~\ref{sec:2vass:flatness}, we prove our main technical result,
namely that for any 2-VASS the global reachability relation can be
characterized by small bounded languages; the latter are also known as
linear path schemes in the literature. Section~\ref{geometry} is
devoted to proving our main theorem, namely that reachability is
$\PSPACE$-complete. We also discuss some further corollaries and
implications of our results there. Finally, we conclude in
Section~\ref{sec:conclusion}, where we discuss open problems and
directions for future work.

\section{Preliminaries}\label{prelims}

In this section, we provide definitions relevant to this paper and
introduce vector addition systems.

\medskip

\noindent
{\bf General notation. } By $\N = \{0,1,2,\ldots\}$, $-\N =
\{0,-1,-2,\ldots\}$ and $\Z$ we denote the sets of non-negative
integers, non-positive integers and integers, respectively.  By $\Q$
and $\Q_{\ge 0}$ we denote the set of rationals and non-negative
rationals, respectively. We define $[i,j] \defeq \{i, i+1, \ldots,
j\}$ for any $i,j \in \Z$.  For each $k\in\Z$ we write $[k,\infty)$ to
  denote $\{z\in\Z : z\geq k\}$.  A {\em quadrant} is one of the four
  sets $\N^2, -\N \times \N, \N \times -\N$ and $-\N \times -\N$.
  Given two vectors $\vec{u} = (u_1,\ldots,u_d), \vec{v} =
  (v_1,\ldots,v_d) \in \Z^d$, we denote by $\vec{u} + \vec{v} \defeq
  (u_1+v_1, \ldots, u_d+v_d)$ the sum of their components. Given two
  sets $U,V \subseteq \Z^d$, we let $U+V \defeq \{\vec{u} + \vec{v} :
  \vec{u} \in U, \vec{v} \in V\}$. The {\em norm} of a vector
  $\vec{u}=(u_1,\ldots,u_d)$ is defined as
  $\norm{\vec{u}}\defeq\max\{|u_i| : i\in[1,d]\}$.  The {\em norm} of
  a matrix $A=(a_{ij})\in\Z^{m\times n}$ is defined as $\norm{A}\defeq
  n\cdot \max\{|a_{ij}|:i\in[1,m], j\in[1,n]\}$. For any word $w =
  a_1\cdots a_n \in \Sigma^n$ over some alphabet $\Sigma$, $w[i,j]$
  denotes $a_i a_{i+1} \cdots a_j$ for all $i,j\in[1,n]$.

\medskip

\noindent
{\bf Graphs, Parikh Images and Linear Path Schemes. } For each set
$\Sigma$, a {\em $\Sigma$-labeled directed graph} is a pair $G=(U,E)$,
where $U$ is a set of {\em vertices} and $E \subseteq U \times \Sigma
\times U$ is a set of {\em edges}. We say $G$ is {\em finite} if $U$
and $E$ are finite. Let $\pi = (u_1,a_1,u_1') \cdots (u_k,a_k,u_k')
\in T^k$.  The {\em Parikh image} $\Par_\pi$ of $\pi$ is the mapping
from $\Sigma$ to $\N$ such that $\Par_\pi(a) = |\{i \in [1,k] : a_i =
a\}|$ for each $a \in \Sigma$. If $X \subseteq E^*$, then $\Par_X$
denotes the set of Parikh images of $X$, \ie $\Par_X = \{\Par_\pi :
\pi\in X\}$.  We say $\pi$ is a {\em path (from $u_1$ to $u_k'$)} if
$u_{i}' = u_{i+1}$ for all $i \in [1,k-1]$.  A path $\pi$ is a {\em
  cycle} if $k \geq 1$ and $u_1 = u_k'$, and {\em cycle-free} if no
infix of $\pi$ is a cycle. A cycle $\pi$ is called {\em simple} if
$\pi$ is the only infix of $\pi$ that is a cycle. A {\em linear path
  scheme} (from $u \in U$ to $u' \in U$) is a regular expression
(whose language will be referred to implicitly) of the form
\begin{align*}
  \rho = \alpha_0 \beta_1^* \alpha_1 \cdots \beta_k^* \alpha_k,
\end{align*} 
where $\alpha_0 \beta_1 \alpha_1 \cdots \beta_k \alpha_k$ is a path
(from $u$ to $u'$) and each $\beta_i$ is a cycle.  We define its {\em
  length} as $|\rho| \defeq |\alpha_0 \beta_1 \alpha_1 \cdots \beta_k
\alpha_k|$. We call $\beta_1, \ldots, \beta_k$ the \emph{cycles of
  $\rho$}.  Note that every path is a linear path scheme. The general
structure of a linear path scheme is illustrated in
Figure~\ref{fig:lps}.

\medskip

\noindent
{\bf Vector Addition Systems with States. } A \emph{vector addition
  system with states (VASS)} in dimension $d$ ($d$-VASS for short) is
a finite $\Z^d$-labeled directed graph $V=(Q,T)$, where $Q$ will be
referred to as the {\em states} of $V$, and where $T$ will be referred
to as {\em transitions} of $V$. The {\em size of $V$} is defined as
$|V| \defeq |Q| + |T| \cdot d \cdot \lceil \log_2 \norm{T}\rceil$,
where $\norm{T}$ denotes the absolute value of the largest number that
appears in $T$, i.e. $\norm{T}\defeq \max\{\norm{\vec{z}} :
(p,\vec{z},q) \in T\}$. We say that $V$ is encoded in \emph{binary}
when we use this definition of $|V|$, which we will use as standard
encoding in this paper. Alternatively, when we set $|V| \defeq |Q| +
|T| \cdot d \cdot \norm{T}$ we say that $V$ is encoded in
\emph{unary}.

Subsequently, $Q \times \Z^d$ denotes the set of {\em configurations}
of $V$. Note that in the literature, the set of configurations is
usually $Q\times \N^d$, however in this paper we will often deal with
VASS whose counters can take integer values. For the sake of
readability, we write configurations $(q,(z_1,\ldots,z_d))$ and
$(q,\vec{z})$ as $q(z_1,\ldots,z_d)$ and $q(\vec{z})$, respectively.

For every subset $\A\subseteq\Z^d$, $p(\vec{u}), q(\vec{v}) \in Q
\times\A$ and every transition $t=(p,\vec{z},q)$, we write $p(\vec{u})
\xrightarrow{t}_{\A} q(\vec{v})$ whenever $\vec{v} = \vec{u} +
\vec{z}$. We extend $\xrightarrow{t}_{\A}$ to sequences of transitions
$\pi \in T^*$ as follows: $\xrightarrow{\pi}_{\A}$ is the smallest
relation satisfying the following conditions for all configurations
$p(\vec{u})$, $q(\vec{v})$, $r(\vec{w})\in Q\times \A$ and all $t \in
T$,
\begin{itemize}
\item $p(\vec{u}) \xrightarrow{\varepsilon}_{\A} p(\vec{u})$
  and
\item if $p(\vec{u}) \xrightarrow{\pi}_{\A} q(\vec{v})$ and
  $q(\vec{v}) \xrightarrow{t}_{\A} r(\vec{w})$, then $p(\vec{u})
  \xrightarrow{\pi \cdot t}_{\A} r(\vec{w})$.
\end{itemize}
We extend $\xrightarrow{t}_\A$ to languages $L \subseteq T^*$ in the
natural way, $\xrightarrow{L}_\A \defeq \bigcup\{\xrightarrow{\pi}_\A
: \pi \in L\}$.  We write $\xrightarrow{*}_{\A}$ to denote
$\xrightarrow{T^*}_\A$.  An {\em $\A$-run} from $q_0(\vec{v}_0)\in
Q\times\A$ to $q_k(\vec{v}_k)\in Q\times\A$ that is induced by a path
$\pi=t_1\cdots t_k$ is a sequence of configurations
$q_0(\vec{v}_0)\xrightarrow{t_1}_{\A}q_1(\vec{v}_1)\cdots\xrightarrow{t_k}_{\A}q_k(\vec{v}_k)$
that we sometimes just abbreviate by
$q_0(\vec{v}_0)\xrightarrow{\pi}_{\A}q_k(\vec{v}_k)$.  When $\A=\N^d$
we also refer to an $\A$-run as a {\em run}.

In the remainder of this paper, we call $\xrightarrow{*}_{\N^d}$ the
{\em reachability relation}, and $\xrightarrow{*}_{\Z^d}$ the {\em
  $\Z$-reachability relation}.  Let $\pi = (p_1,\vec{z}_1,p_1) \cdots
(p_k,\vec{z}_k,p_k) \in T^k$ for some $k \geq 0$. The
\emph{displacement of $\pi$} is $\delta(\pi) \defeq \sum_{i=1}^k
\vec{z}_i$, and the definition naturally extends to languages $L
\subseteq T^*$ as $\delta(L) \defeq \bigcup\{\delta(\pi) : \pi \in
L\}$.  We say that a linear path scheme $\rho$ over $V$ {\em captures}
a linear path scheme $\rho'$ if $\delta(\rho')\subseteq\delta(\rho)$.
Note in particular that if $\Par_{\rho'}\subseteq\Par_\rho$, then
$\delta(\rho')\subseteq\delta(\rho)$. Similarly as in~\cite{LS04}, we
say that a linear path scheme $\alpha_0 \beta_1^* \alpha_1 \cdots
\beta_k^* \alpha_k$ is {\em zigzag-free} if $\{\delta(\beta_1),
\ldots, \delta(\beta_k)\} \subseteq Z$ for some quadrant $Z$.

\begin{figure}[t]
  \includegraphics{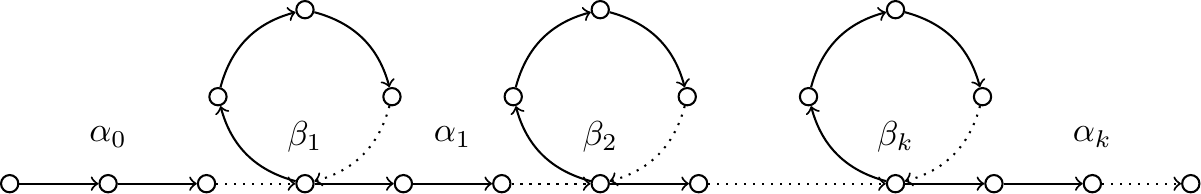}
  \caption{\label{fig:lps} Illustration of the structure of a linear
    path scheme $\rho = \alpha_0 \beta_1^* \alpha_1 \cdots \beta_k^*
    \sigma_k$.}
\end{figure}

\section{Main Results}\label{sec:results}

In this paper, our main interest is in the reachability problem for
$2$-VASS, formally defined as follows: \\

\problemx{2-VASS Reachability} {A $2$-VASS $V=(Q,T)$ and
  configurations $p(\vec{u})$ and $q(\vec{v})$ from $Q \times
  \mathbb{N}^2$.}  {Is there a run from $p(\vec{u})$ to $q(\vec{v})$, i.e. does $p(\vec{u}) \xrightarrow{*}_{\N^2}
  q(\vec{v})$ hold?}\

In order to determine the complexity of this problem, we show that the
reachability relation of any 2-VASS can be defined by a finite union
of linear path schemes. In particular, we are able to show strong
bounds on their lengths and their number of cycles. For example,
consider the 2-VASS $V$ depicted in Fig.~\ref{fig:ex:2vass}.
\begin{figure}[h]
  \includegraphics{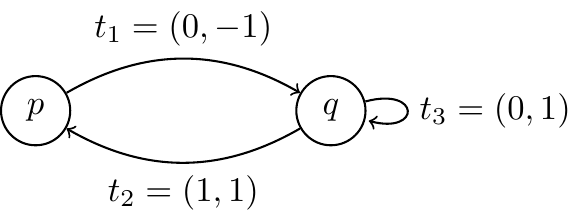}
    
    
  \caption{\label{fig:ex:2vass} Example of a 2-VASS.}
\end{figure}
Since $V$ contains nested loops, e.g. $(t_1 t_3^* t_2)^*$, we cannot
directly read off a characterization of its reachability set by a
finite union of linear path schemes. However, by carefully unraveling
loops we obtain the reachability set from the union of the subsequent
linear path schemes, and in particular this means that $V$ can be
flattened:
\begin{align*}
  p(u_1, u_2) \xrightarrow{*}_{\N^2} q(v_1, v_2) & \iff
  p(u_1, u_2) \xrightarrow{\mathmakebox[2.2cm]{\substack{t_1 t_3^* \;
        \cup \; \\ t_1 t_3^* t_2 (t_1 t_2)^* t_1}}}_{\N^2}
  q(v_1, v_2) & (\text{cf. Fig.~\ref{fig:ex:flat:2vass:pq}})\\
  p(u_1, u_2) \xrightarrow{*}_{\N^2} p(v_1, v_2) & \iff
  p(u_1, u_2)  \xrightarrow{\mathmakebox[2.2cm]{t_1 t_3^* t_2 (t_1
      t_2)^* \; \cup \; \varepsilon}}_{\N^2}  p(v_1, v_2)\\
  q(u_1, u_2) \xrightarrow{*}_{\N^2} p(v_1, v_2) & \iff
  q(u_1, u_2)  \xrightarrow{\mathmakebox[2.2cm]{(t_2 t_1)^* t_3^*
      t_2}}_{\N^2}  p(v_1, v_2)\\
  q(u_1, u_2) \xrightarrow{*}_{\N^2} q(v_1, v_2) & \iff
  q(u_1, u_2) \xrightarrow{\mathmakebox[2.2cm]{(t_2 t_1)^*
      t_3^*}}_{\N^2}  q(v_1, v_2)\qquad .
\end{align*}
We will show that such a flattening exists for any 2-VASS. More
precisely, our main technical result states that the global
reachability relation of any 2-VASS $V=(Q,T)$ can be defined via a
union of linear path schemes whose lengths can be polynomially bounded
in $|Q|+\norm{T}$, and \emph{a fortiori} are at most exponential in
$|V|$, and whose number of cycles is quadratic in $|Q|$:

\begin{theorem}\label{main1}
  Let $V=(Q,T)$ be a 2-VASS. There is a finite set $S$ of linear path
  schemes such that\footnote{The expanded technical meaning of this
    statement is that there are constants $c_1$ and $c_2$ such that
    for every 2-VASS $V=(Q,T)$ there exists a finite set $S$ of linear
    path schemes with the properties that $p(\vec{u})
    \xrightarrow{*}_{\N^2} q(\vec{v})$ if, and only if, $p(\vec{u})
    \xrightarrow{S}_{\N^2} q(\vec{v})$ and that each $\rho$ in $S$ has
    length at most $(|Q| + \norm{T})^{c_1}$ and has at most $c_2|Q|^2$
    cycles.  The more familiar statements of this theorem and of
    lemmas of a similar nature in the rest of the paper were chosen to
    avoid clutter and to downplay the role of the precise constants.
  }
  \begin{itemize}
  \item $p(\vec{u}) \xrightarrow{*}_{\N^2} q(\vec{v})$ if, and only
    if, $p(\vec{u}) \xrightarrow{S}_{\N^2} q(\vec{v})$,
  \item $|\rho| \leq (|Q| + \norm{T})^{O(1)}$ for every $\rho \in S$,
    and
  \item each $\rho \in S$ has at most $O(|Q|^2)$ cycles.
  \end{itemize}
\end{theorem}

Having established Theorem~\ref{main1}, we can show that proving the
existence of a path between two reachable configurations in a 2-VASS
reduces to checking the existence of a solution for suitably
constructed systems of linear Diophantine inequalities that depend on
$S$ and the properties listed in Theorem~\ref{main1}. The absence of
nested cycles in linear path schemes in $S$ is crucial to this
reduction. By application of standard bounds from integer linear
programming, this in turn enables us to bound the length of paths
witnessing reachability, and to prove the main theorem of this paper
in Section~\ref{geometry}:
\begin{theorem}\label{main2}
  {\sc 2-VASS Reachability} is $\PSPACE$-complete.
\end{theorem}

\begin{figure}[t]
  \includegraphics{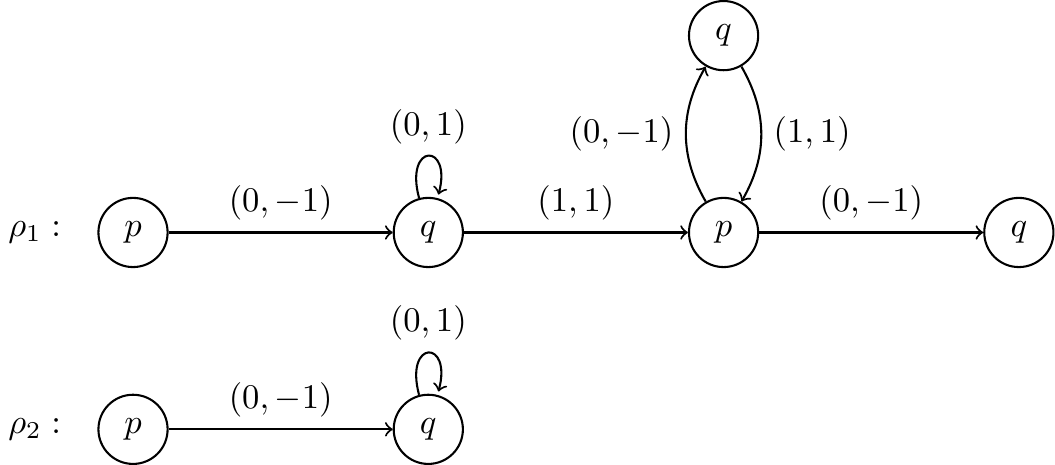}
    

    


  \caption{\label{fig:ex:flat:2vass:pq} Illustration of a set $S =
    \{\rho_1, \rho_2\}$ of linear path schemes defining the
    reachability relation from $p$ to $q$ of the 2-VASS $V$ depicted
    in Fig.~\ref{fig:ex:2vass}. Here, $\rho_1 = t_1 t_3^* t_2 (t_1
    t_2)^* t_1$, $\rho_2 = t_1 t_3^*$, and $p(u_1, u_2)
    \xrightarrow{*}_{\N^2} q(v_1, v_2)$ if, and only if, $p(u_1, u_2)
    \xrightarrow{S}_{\N^2} q(v_1, v_2)$.}
\end{figure}

\section{Proof of Theorem 1}\label{sec:2vass:flatness}

In this section, we prove Theorem~\ref{main1} and show that runs of a
2-VASS $V=(Q,T)$ are captured by a finite union of linear path schemes
each of which have length at most $(|Q|+\norm{T})^{O(1)}$ and at most
$O(|Q|^2)$ cycles.  In order to construct this finite set of linear
path schemes, we consider the following three types of runs
$p(u_1,u_2) \xrightarrow{\pi}_{\N^2} q(v_1,v_2)$, depicted in
Fig.~\ref{fig:types:paths}:
\begin{enumerate}
  \item Both counter values of $p(u_1,u_2)$ and of $q(v_1,v_2)$ are
    sufficiently large and $p=q$, but intermediate configurations on
    the run $p(u_1,u_2) \xrightarrow{\pi}_{\N^2} q(v_1,v_2)$ may have
    arbitrarily small counter values.
  \item For all configurations of the run $p(u_1,u_2)
    \xrightarrow{\pi}_{\N^2} q(v_1,v_2)$ both counter values are
    sufficiently large.
  \item For all configurations of the run $p(u_1,u_2)
    \xrightarrow{\pi}_{\N^2} q(v_1,v_2)$ at least one counter value is
    not too large.
\end{enumerate}
In Sections~\ref{sec:zreach}, \ref{sec:ultimately-flatness} and
\ref{sec:bounded-component}, we will show how to construct linear path
schemes for these three types of runs. Then, in
Section~\ref{sec:glue}, we prove Theorem \ref{main1} by showing that
any run can be decomposed as finitely many runs of these types.

In some more detail, the first step is to show in
Section~\ref{sec:zreach} that Parikh images of finite labeled graphs
can be captured by linear path schemes of polynomial size. This will
allow us to prove that $\Z$-reachability, i.e.\ runs in which counter
values may drop below zero, can be captured by linear path schemes of
polynomial size. We then give in Section~\ref{sec:ultimately-flatness}
an effective decomposition of certain linear sets in dimension two
into semi-linear sets with special properties, and use this
decomposition in order to derive together with the results in
Section~\ref{sec:zreach} linear path schemes of size
$(|Q|+\norm{T})^{O(1)}$ with a constant number of cycles for runs of
type~(1). Linear path schemes for runs of type~(2) will then be seen
to follow from the type~(1) case.

For runs of type~(3), in Section~\ref{sec:bounded-component} we
construct linear path schemes for $1$-VASS and show that runs of a
$2$-VASS that stay within an ``L-shaped band'' are, essentially, runs
of a $1$-VASS.  Our analysis of such runs of type~(3) is a simple
consequence of certain normal forms of shortest runs in one-counter
automata, which $1$-VASS are a subclass of, by Valiant and Paterson
\cite{VP75}.

\begin{figure}[t]
  \begin{center}
    \begin{tabular}{cc}
      \includegraphics{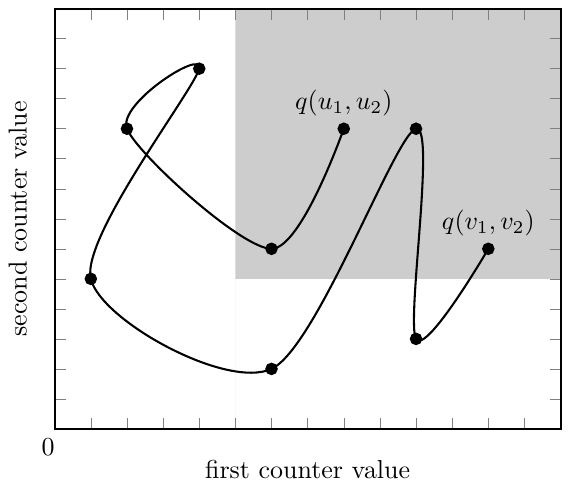}

          
        
          
          
      &
      \includegraphics{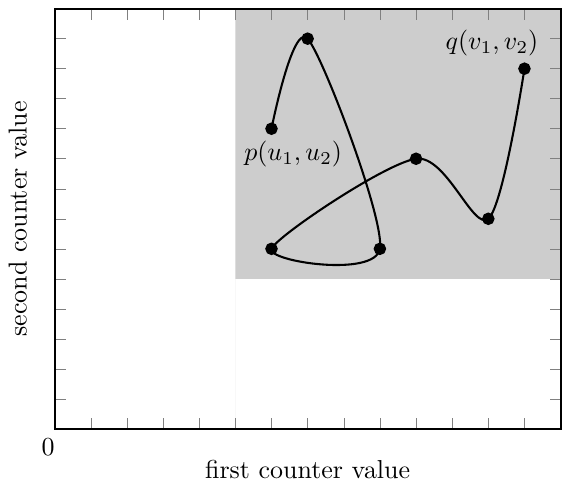}
          
        
          
          
      \\[5pt]
      
      \multicolumn{2}{c}{
      \includegraphics{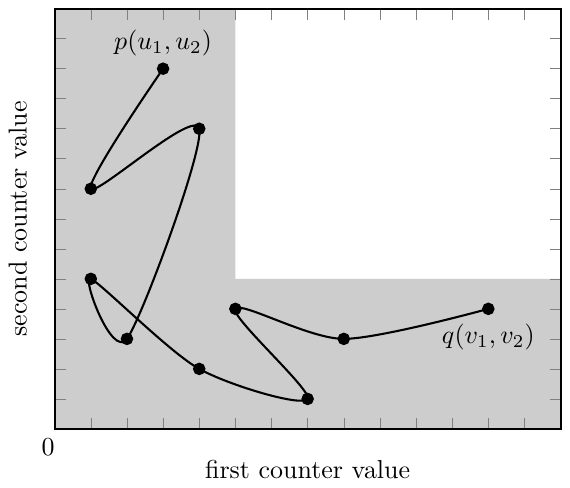}
          
          
        
          
          
      } \\
    \end{tabular}
  \end{center}
  \caption{Example of the three types of runs. The region depicted in
    each case is the positive quadrant in the Cartesian plane. (1)
    top-left: run from $q$ to $q$ starting and ending sufficiently
    high; (2) top right: run staying sufficiently high; (3) bottom:
    run within an L-shaped band, i.e., running high on at most one
    component at a time.\label{fig:types:paths}}
\end{figure}

\medskip

\noindent
{\bf Similarities and differences in comparison with \cite{LS04}.\ }
Our proof strategy of considering the three kinds of runs described
above shares some similarities with~\cite{LS04}, but in particular
requires to explicate all implicit assumptions made in the conference
paper~\cite{LS04}. There, the bounds on what is referred to as
``large'' and what is referred to as ``not large'' or ``small'' in the
runs of type~(1), (2) and (3) are not explicitly calculated. Our
proofs for obtaining rather tight bounds require new insights. We
capture runs of type (1) by linear path schemes of size
$(|Q|+\norm{T})^{O(1)}$, whereas in \cite{LS04} the linear path
schemes were of size at least exponential in $|Q|$. To prove the
former, we establish a new upper bound on the presentation size of
Parikh images of finite automata in Lemma~\ref{L parikh} below, which
is a result of independent interest. The difference between our runs
of type~(2) and the ones analyzed in \cite{LS04} is that our runs have
to stay in the ``outside region'' entirely, whereas in \cite{LS04} the
set of displacements of paths from $q$ to $q'$ is analyzed. Runs of
type~(3) are treated as special cases of runs of type~(2) in
\cite{LS04}, whereas we invoke a result by Valiant and Paterson on
normal forms of minimal runs in one-counter automata. Our final proof
of Theorem \ref{main1} shows that each run can be factorized into
segments of runs of types~(1), (2) and (3) and requires a more careful
treatment than in \cite{LS04}. At every step, we have to ensure that
the number of cycles of the linear path schemes we construct stays
polynomial in the number of control states $Q$.  This aspect is
neglected in~\cite{LS04} as it is of no interest for the goal
of~\cite{LS04}, however, for us it is by far the technically most
challenging part and one of the cornerstones of our PSPACE upper
bound.

\subsection{Parikh images of finite directed graphs and $\Z$-reachability of $d$-VASS}\label{sec:zreach}

The main result of this section is the following proposition.

\begin{proposition}\label{P zreach}
  Let $V = (Q, T)$ be a $d$-VASS. There exists a finite set $S$ of
  linear path schemes such that
  \begin{itemize}
  \item[(i)] $p(\vec{u}) \xrightarrow{*}_{\Z^d} q(\vec{v})$ if, and only
    if, $p(\vec{u}) \xrightarrow{S}_{\Z^d} q(\vec{v})$,
  \item[(ii)] $|\rho|\leq 2\cdot |Q|\cdot|T|$ for each $\rho \in S$,
    and
  \item[(iii)] each $\rho\in S$ has at most $|T|$ cycles.
  \end{itemize}
\end{proposition}

In order to prove Proposition \ref{P zreach}, we will prove suitable
bounds on the representation size of the Parikh images of paths of a
$\Sigma$-labeled finite graphs (or equivalently, nondeterministic
finite automata) in terms of linear path schemes.

\begin{lemma}\label{L parikh}
  Let $G=(U,E)$ be a finite $\Sigma$-labeled graph. There exists a
  finite set $S$ of linear path schemes such that
  \begin{itemize}
  \item[(i)] $\{\Par_\pi : \pi \text{ is a path}\} =
    \bigcup\{\Par_\rho : \rho \in S\}$,
  \item[(ii)] $|\rho|\leq 2\cdot|U|\cdot|E|$ for each $\rho\in S$, and
  \item[(iii)] each $\rho \in S$ has at most $|E|$ cycles.
  \end{itemize}
\end{lemma}

\begin{proof}
  We first provide some additional definitions. Let
  $\sigma,\sigma':E\rightarrow\N$ be mappings and let $X$ be a set of
  such mappings. We define $\sigma+\sigma'\in\N^E$ as
  $(\sigma+\sigma')(e) \defeq \sigma(e)+\sigma'(e)$ for each $e \in E$
  and $X+\sigma\defeq\{\tau+\sigma:\tau\in X\}$.  For each $u\in U$,
  let $\text{in}(u)\defeq\{(u',a,u'')\in E: u''=u\}$ and and
  $\text{out}(u)\defeq\{(u',a,u'')\in E\mid u'=u\}$ denote the set of
  incoming and outgoing edges of $u$, respectively. We say that
  $\sigma$ is \emph{flow-preserving} if for every $u \in U$ we have
$$\sum_{e\in \text{in}(u)}
  \sigma(e) \qquad=\qquad \sum_{e\in \text{out}(u)}\sigma(e)\quad.$$

\noindent
  We will show the following claim: \newline
  
  {\em Claim.} Let $\pi \in E^*$ be a path. There exists some $h\geq
  1$, a sequence of linear path schemes $\rho_1, \ldots, \rho_h \subseteq E^*$, and a
  sequence $\sigma_1, \ldots, \sigma_h \in \N^E$ such that
  \begin{enumerate}[(a)]
    \item $\rho_1$ is a path of length at most $|U|\cdot|E|$ that visits each state of
$\pi$ at least once,
    \item $\sigma_1$ is flow-preserving, and
    \item $\Par_\pi = \Par_{\rho_1} + \sigma_1$, 
  \end{enumerate}
  and for every $1 < i \leq h$,
  \begin{enumerate}[(1)]
  \item $\rho_i$ is a linear path scheme that can be obtained from $\rho_1$ by
    inserting $i-1$ simple cycles (of the form $\beta^*$),
 \item $\sigma_i$ is flow-preserving,
  \item $\Par_{\rho_{i-1}} + \sigma_{i-1} \subseteq
    \Par_{\rho_i} + \sigma_i$,
  \item $\sigma_{i-1}(e)\geq\sigma_i(e)$ for all $e\in E$ and 
there exists some $e\in E \text{ s.t. } \sigma_{i-1}(e) > \sigma_i(e) = 0$, and
  \item $\sigma_h(e) = 0$ for all $e\in E$.
  \end{enumerate}

First observe that due to (4) we have $h \leq |E|$, and due to (1) we
have $|\rho_i|\leq|\rho_1|+|U|\cdot(i-1)$. Therefore $|\rho_h| \leq
|\rho_1| + |U|\cdot|E| \leq 2 \cdot |U|\cdot|E|$, where the last
inequality is due to (a). Moreover $\rho_h$ has at most $|E|$ cycles
due to (1) and $h \leq |E|$.

Before proving the claim, let us first see how it proves the lemma.
We define $$S \quad \defeq \quad \{\rho:\rho \text{ is a linear path scheme, }
|\rho| \leq 2\cdot|U|\cdot|E|\text{ and $\rho$ has at most $|E|$
  cycles}\}\quad.
$$ Trivially, (ii) and (iii) are satisfied. To establish (i) let us
fix an arbitrary path $\pi$. We apply the above Claim and obtain a linear
path scheme $\rho_h\in S$ for $\pi$. It suffices to show $\Par_\pi \in
\Par_{\rho_h}$ which holds due to
  $$ \Par_\pi \stackrel{\text{(c)}}{=} \Par(\rho_1) +
  \sigma_1 \stackrel{\text{(3)}}{\subseteq} \quad \cdots
  \quad \stackrel{\text{(3)}}{\subseteq} \Par_{\rho_h} +
  \sigma_h \stackrel{\text{(5)}}{=} 
 \Par_{\rho_h}.
  $$ 

  We now prove the claim. Let $\pi$ be a path and let us first define
  $\rho_1$ and $\sigma_1$ such that (a),(b) and (c) are satisfied.
  The path $\pi$ can be decomposed as $\pi = e_1 \pi_1 \cdots e_k
  \pi_k$ where $k \leq |U|$ and each $e_j = (u, a, u')$ is the first
  transition such that $u$ or $u'$ appears in $\pi$.  We define
  $\rho_1$ and $\sigma_1$ as the result of the following iterative
  process: We initially set $\rho_1$ to $\pi$ and set $\sigma_1(e)=0$
  for all $e\in E$; then we successively remove a simple cycle $\beta$
  from some $\pi_j$, and add $\Par(\beta)$ to $\sigma_1$.  We repeat
  this process until no longer possible. The resulting $\rho_1$ is a
  path of length at most $|U|\cdot|E|$.  Moreover, $\sigma_1$ is
  flow-preserving since we successively removed cycles only, and
  clearly $\Par(\pi) = \Par(\rho_1) + \sigma_1$, by construction. Thus
  (a), (b) and (c) hold.

  Let us prove (1) to (5) by induction on $1 < i \leq h$. We only
  prove the induction step, the base case can be proven
  analogously. Let $E' \defeq \{e \in E : \sigma_{i-1}(e)>0\}$. If $E'
  = \emptyset$, then (5) holds and we are done. Thus, we assume that
  $E' \not= \emptyset$. Let us fix a choice function $\chi : E'
  \rightarrow E'$ satisfying $$ \chi(u_1,a,u_2)=(u_1',a,u_2') \qquad
  \Longrightarrow \qquad u_2 = u_1'.$$ Note that $\chi$ exists since
  $\sigma_{i-1}$ is flow-preserving by induction hypothesis. By the
  pigeonhole principle there exist some $e\in E'$ and some $\ell\geq
  0$ such that $\beta\defeq e\chi(e)\chi^2(e)\cdots\chi^\ell(e)$ is a
  simple cycle and $c\defeq\sigma_{i-1}(e)\leq\sigma_{i-1}(\chi^j(e))$
  for all $j\in[1,\ell]$. We define $\sigma_i \defeq \sigma_{i-1} -
  \Par(\beta^c)$ and observe that $\sigma_i$ is flow-preserving due to
  minimality of $c$; thus (2) and (4) are shown. Let $\beta$ be a
  cycle from $u$ to $u$. By (1) of induction hypothesis the linear
  path scheme $\rho_{i-1}$ can be obtained from $\rho_1$ by inserting
  $(i-2)$ simple cycles and can hence be factorized as $\rho_{i-1} =
  \alpha\gamma$, where $\alpha$ is a linear path scheme from some
  state to $u$. We set $\rho_i\defeq\alpha\beta^*\gamma$ and hence (1)
  holds. Furthermore, (3) holds due to $\Par_{\rho_{i-1}} +
  \sigma_{i-1} = \Par_{\rho_{i-1}} + \Par_{\beta^c} + \sigma_i
  \subseteq \Par_{\rho_i} + \sigma_i$.
\end{proof}

We are now prepared to prove Proposition \ref{P zreach}.

\begin{proof}[Proof of Proposition \ref{P zreach}]
  We have $T\subseteq Q \times \Sigma \times Q$ for some finite subset
  $\Sigma \subseteq \Z^d$. Let $S$ be the finite union of linear path schemes from
  Lemma \ref{L parikh}, then (2) and (3) are clear. For (1) we have
  the following equivalences: 
  \begin{eqnarray*}
    p(\vec{u})\xrightarrow{*}_{\Z^d}q(\vec{v})
    &\Longleftrightarrow&
    \exists\text{ path }\pi\text{ from $p$ to $q$ in $V$ s.t. }
    \vec{v}-\vec{u}=\sum_{\vec{z}\in\Sigma}\Par_\pi(\vec{z})\cdot\vec{z} \\[5pt]
    &\stackrel{\text{Lemma \ref{L parikh}~(i)}}{\Longleftrightarrow}& 
    \text{$\exists \rho\in S$ 
      from $p$ to $q$, $\exists f\in\Par_\rho$ s.t. }
    \vec{v}-\vec{u}\in\sum_{\vec{z}\in\Sigma}f(\vec{z})\cdot\vec{z} \\[5pt]
    &\Longleftrightarrow&
    \text{$\exists\rho\in S$ from $p$ to $q$ s.t. }\vec{v}-\vec{u}\in\delta(\rho) \\[5pt]
    &\Longleftrightarrow&
    p(\vec{u})\xrightarrow{S}_{\Z^d}q(\vec{v})
  \end{eqnarray*}
\end{proof}

\subsection{Starting and ending in ``sufficiently large'' configurations}
\label{sec:ultimately-flatness}

The goal of this section is to prove that, given a 2-VASS $V$, there
exists a sufficiently small bound $D$ such that the reachability
relation between any two configurations $q(u_1,v_1)$ and $q(u_2,v_2)$
for which $u_1,u_1,u_2,v_2 \geq D$ can be captured by a finite set of
small linear path schemes (in the sense of
Theorem~\ref{main1}). In~\cite{LS04}, this property is referred to as
\emph{ultimately flat}.  As a consequence of this result, we can show
that the reachability relation between arbitrary configurations for
which there exists a run on which both counter values on all
configurations stay above $D$ can be captured by a finite union of
small linear path schemes as well.
\begin{proposition}
  \label{P upper right}
  Let $V=(Q,T)$ be a 2-VASS. There exist $D\le (\abs{Q} +
  \norm{T})^{O(1)}$ and sets of linear path schemes $R,X$ such that
  for $\O\defeq [D, \infty)^2$ and $\vec{u},\vec{v}\in \O$,
    \begin{enumerate}[(a)]
    \item \begin{itemize}
    \item $q(\vec{u}) \xrightarrow{*}_{\N^2} q(\vec{v})$ if, and only
      if, $q(\vec{u}) \xrightarrow{R}_{\N^2} q(\vec{v}) $, and
    \item $|\rho| \leq (|Q| + \norm{T})^{O(1)}$ and $\rho$ has at most two cycles
      for every $\rho\in R$.
    \end{itemize}
    \item \begin{itemize}
    \item $p(\vec{u}) \xrightarrow{*}_\O q(\vec{v})$ if, and only if,
      $p(\vec{u}) \xrightarrow{X}_{\N^2} q(\vec{v})$, and
    \item $|\rho| \leq (|Q| + \norm{T})^{O(1)}$ and $\rho$ has at most
      $2\cdot|Q|$ cycles for every $\rho \in X$.
    \end{itemize}
    \end{enumerate}
\end{proposition}
The proof of this proposition requires two intermediate steps. First,
in Lemma~\ref{L Dimension Two} below we prove an effective
decomposition of certain linear sets in dimension two into semi-linear
sets with nice properties. Similar decompositions have been the
cornerstone of the results by Hopcroft and Pansiot~\cite{HP79} and
Leroux and Sutre~\cite{LS04}. The contribution of Lemma~\ref{L
  Dimension Two} is to establish a new proof from which we can obtain
sufficiently small bounds on this decomposition. Next, in
Lemma~\ref{lem:zig-zag-free-lps} we show how this decomposition can be
applied in order to capture reachability instances by linear path
schemes with two cycles whose displacements all point into the same
quadrant. This in turn enables us to prove Part~(a) of
Proposition~\ref{P upper right}, from which we can then prove
Part~(b).

Let us recall some definitions concerning semi-linear sets. Let $P =
\{\vec{p}_1,\ldots, \vec{p}_n\} \subseteq \Z^m$ and $\D \subseteq
\Q$. The {\em $\D$-cone generated by $P$} is defined as
\begin{align*}
  \cone_\D(P) \defeq \left\{\sum_{i \in [1,n]} \lambda_i \cdot
  \vec{p}_i : \lambda_i \in \D, \lambda_i \geq 0\right\}.
\end{align*}
A {\em linear set} $L(\vec{b}; P)$ is given by a {\em base vector}
$\vec{b} \in \Z^d$ and a finite set of {\em period vectors} $P
\subseteq \Z^d$, where $L(\vec{b}; P) \defeq \vec{b}+\cone_\N(P)$. A
          {\em semi-linear set} is a finite union of linear sets. The
          {\em norm $\norm{P}$} of a finite set $P \subseteq \Z^d$ is
          defined as $\norm{P}\defeq \max\{\norm{\vec{p}} : \vec{p}\in
          P\}$. Recall that $\vec{u},\vec{v} \in \Z^d$ are {\em
            linearly dependent} if $\vec{0} = \lambda_1 \cdot \vec{u}
          + \lambda_2 \cdot \vec{v}$ for some $\lambda_1,\lambda_2 \in
          \Q \setminus \{0\}$, and {\em linearly independent}
          otherwise.

We now show the following statement: the intersection of a linear set
$L(\vec{b}; P) \subseteq \Z^2$, such that $\vec{b} \in P$, with 
some quadrant $Z$ is equal to a semi-linear set
$\bigcup_{i \in I} L(\vec{c}_i; P_i)$ such that $\vec{c}_i \in
L(\vec{b}; P)$ and each $P_i$ contains only two ``small'' vectors from
$(P\cup L(\vec{b}; P)) \cap Z$.

\begin{lemma}{\label{L Dimension Two}}
  Let $\vec{b} \in \Z^2$, let $P \subseteq \Z^2$ be finite with
  $\vec{b} \in P$ and let $Z$ be a quadrant. Then $L(\vec{b}; P) \cap Z = \bigcup_{i
    \in I} L(\vec{c}_i; P_i)$ such that for each $i \in I$ we have
  \begin{itemize}
  \item $|P_i| \leq 2$,
  \item $P_i \subseteq (P \cup L(\vec{b}; P)) \cap Z$,
    and
  \item there exists $e\leq\norm{P}^{O(1)}$ such that
    $\{\vec{c}_i\}\cup(P_i\cap
    L(\vec{b};P))\subseteq\vec{b}+\cone_{[0,e]}(P)$.
\end{itemize}
\end{lemma}

\begin{proof}
  Subsequently, we assume that $Z=\N^2$ and that $P$ only contains pairwise linearly
  independent vectors, the general case can be obtained as an adaption
  of our argument. Let $P=\{\vec{p}_1,\ldots, \vec{p}_n\}$ and
  $\vec{r}\in L(\vec{b};P) \cap \N^2$, by definition $\vec{r} =
  \vec{b} + \lambda_1 \vec{p}_1 + \cdots + \lambda_n \vec{p}_n$ for
  some $\lambda_i\in\N$. Denote $B\defeq \norm{P}$ and suppose there
  are more than two $\lambda_i$ greater than $B^2$, say
  $\lambda_1,\lambda_2,\lambda_3> B^2$. An easy calculation shows that
  there exist $\gamma_1 \in [1,B^2]$  and
  $\gamma_2,\gamma_3\in [-B^2,B^2]$ such that $\gamma_1 \vec{p}_1 =
  \gamma_2 \vec{p}_2 + \gamma_3 \vec{p}_3$. We can thus always
  decrease all but two $\lambda_i$ below $B^2$.
 Hence we can write
  $L(\vec{b};P)$ as the following semi-linear set whose base vectors
  are sufficiently small and whose period vectors have cardinality at
  most two, where  $W\defeq\{\sum_{i=1}^n\lambda_i\vec{p}_i : \lambda_i\in [0,B^2]\}$:
  \begin{align*}
    L(\vec{b};P)  \quad=\quad \bigcup_{P'\subseteq P, \abs{P'}\le 2}\;\bigcup_{\vec{w}\in W} 
    L(\vec{b} + \vec{w};P').
  \end{align*}
  Consequently, $\vec{r} = \vec{z} + \lambda\cdot \vec{u} + \zeta
  \cdot \vec{v}$ where $\vec{z}\defeq \vec{b} + \vec{w}\in L(\vec{b};P)$ for some
  $\vec{w}\in W$. Since $\norm{\vec{w}}\le \norm{P}^{O(1)}$ for every
  $\vec{w}\in W$, we have $\norm{\vec{z}}\le \norm{P}^{O(1)}$.

Our goal is to show that $\vec{r}$ lies in a linear set
  fulfilling the properties required in the lemma. If $\{\vec{u},
  \vec{v}\} \subseteq \N^2$, then we are done since then we have
  $\vec{u},\vec{v}\in P\cap\N^2$.
  The cases when $\lambda=0$ or $\zeta=0$ are trivial, hence we
  subsequently assume $\lambda,\zeta > 0$.  
  We thus consider the remaining
  cases separately up to symmetry. Let $\vec{u}=(u_1,u_2)$ and
  $\vec{v}=(v_1,v_2)$. Note that due to
  $\vec{r}=\vec{z}+\lambda\cdot\vec{u}+\zeta\cdot\vec{v}\in\N^2$ we
  must have $\lambda\cdot \vec{u} + \zeta\cdot \vec{v}\in
  [-\norm{\vec{z}},\infty) \times [-\norm{\vec{z}},\infty)$.

\medskip

\noindent
  {\em Case 1:} $\vec{u} \in \N^2$ and $\vec{v}\not\in \N^2$. We only treat the case when the clockwise angle between
  $\vec{u}$ and $\vec{v}$ exceeds $180^\circ$, illustrated in
  Figure~\ref{fig:case1}, i.e. when $\vec{v} \in -\N_{>0} \times \Z$
  and $u_2/u_1 < v_2/v_1$. The case when the clockwise angle is below
  $180^\circ$ can be treated symmetrically, i.e. when $\vec{v} \in \Z
  \times -\N_{>0}$ and $u_2/u_1 > v_2/v_1$. It cannot be exactly
  $180^\circ$ since $\vec{u}$ and $\vec{v}$ are linearly independent
  by assumption.

    \begin{figure}[t]
      \begin{center}
        \includegraphics{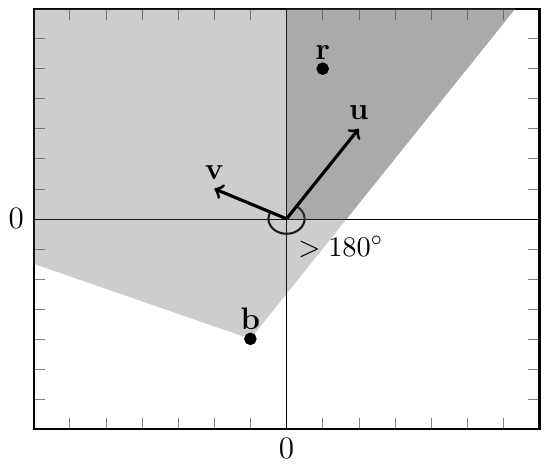}
      \end{center}

      \caption{Example of Case 1 with the angle exceeding
        $180^\circ$. The filled area corresponds to $\vec{b} +
        \cone_{\Q_{\geq 0}}(\{\vec{u}, \vec{v}\})$
        and the darker filled area corresponds to $\vec{b} +
        \cone_{\Q}(\{\vec{u}, \vec{v}\}) \, \cap\ (\Q_{\geq 0}
        \times \Q_{\geq 0})$.
      }
      \label{fig:case1}
    \end{figure}

  Our first step is to show the existence of some $\alpha\in \N$ such
  that $(0,\alpha) = \beta \cdot \vec{u} + \gamma \cdot \vec{v} \in
  L(\vec{b}; P)$ for some $\beta, \gamma \in [1,B^4]$. Due to
  the linear independence of $\vec{u}$ and $\vec{v}$, there exist
  $\eta \in [1,B^2]$ and $\chi, \theta \in [-B^2,B^2]$ such that $\eta
  \cdot \vec{b} = \chi \cdot \vec{u} + \theta \cdot \vec{v}$. Since
  the clockwise angle between $\vec{u}$ and $\vec{v}$ exceeds
  $180^\circ$ there are positive $\alpha',\beta',\gamma' \in [1,B^2]$
  such that $(0,\alpha') = \beta' \cdot \vec{u} + \gamma' \cdot
  \vec{v}$. Thus, we can choose $\alpha,\beta$ and $\gamma$ as
  follows:
  \begin{eqnarray}
    (0,\underbrace{\alpha'\cdot B^2}_{\alpha}) & = &
    \underbrace{B^2\cdot\beta'}_{\beta} \cdot \vec{u} +
    \underbrace{B^2\cdot\gamma'}_{\gamma} \cdot \vec{v} \label{E Up}
    \\ & = & \beta \cdot \vec{u} + \gamma \cdot \vec{v} - \eta \cdot
    \vec{b} + \eta \cdot \vec{b} \nonumber \\ & = &
    (\underbrace{\beta-\chi}_{\geq 0}) \cdot \vec{u} +
    (\underbrace{\gamma-\theta}_{\geq 0}) \cdot \vec{v} + \eta \cdot
    \vec{b} \stackrel{\vec{b} \in P,\eta \geq 1}{\in} L(\vec{b};
    P)\nonumber .
  \end{eqnarray}
  As an intermediate step, we show that
  \begin{eqnarray*}
    \label{E Linear} \gamma \cdot \lambda - \beta \cdot \zeta > -\norm{P}^{O(1)}.
  \end{eqnarray*} 
  To this end we rewrite $\lambda \cdot \vec{u} + \zeta \cdot \vec{v}$
  as
  \begin{eqnarray}
    \lambda \cdot \vec{u} + \zeta \cdot \vec{v} &
    = & \left(\left\lfloor\frac{\lambda}{\beta}\right\rfloor
    \cdot \beta + (\lambda \bmod \beta)\right) \cdot
    \vec{u} + \zeta \cdot \vec{v}\nonumber \\
    & \stackrel{(\ref{E Up})}{=} &
    \left\lfloor\frac{\lambda}{\beta} \right\rfloor
    \left((0,\alpha) - \gamma \cdot \vec{v}\right) + (\lambda
    \bmod \beta) \cdot \vec{u} + \zeta \cdot
    \vec{v}\nonumber \\
    & = & \underbrace{\left(\zeta -
      \left\lfloor\frac{\lambda}{\beta}\right\rfloor \cdot
      \gamma \right)}_{\kappa} \cdot \vec{v} +
    \underbrace{(\lambda \bmod \beta) \cdot
      \vec{u}}_{\text{has norm at most $B^5$}} +
    \left\lfloor\frac{\lambda}{\beta}\right\rfloor
    \cdot(0,\alpha) \label{E Form}.
  \end{eqnarray}
  Recall that $v_1<0$. Since $\lambda\cdot \vec{u} + \zeta\cdot
  \vec{v}\in [-\norm{\vec{z}},\infty) \times \Z$, applying (\ref{E
      Form}) we derive
  \begin{align*}
    \kappa \cdot v_1 + (\lambda \bmod \beta)\cdot \norm{\vec{u}} \ge
    - \norm{\vec{z}}\quad \Longrightarrow\quad \kappa \le H ~~\text{for some } H\le \norm{P}^{O(1)}.
  \end{align*} 
We now obtain
  \begin{eqnarray}
    \zeta - \frac{\lambda + \beta}{\beta}\cdot \gamma < \kappa \le H&\quad\Longrightarrow\quad&\zeta<H+\frac{\lambda+\beta}{\beta}\cdot\gamma
\nonumber\\
& \Longrightarrow&
    \label{E Lambda Ineq}
\lambda > \frac{(\zeta - H) \cdot \beta}{\gamma} - \beta.
  \end{eqnarray} 
  
\noindent
  In order to obtain $\vec{r} \in L(\vec{c}; \{\vec{x}, \vec{y}\})$
  for suitable $\vec{c}, \vec{x}, \vec{y}$, we make a case
  distinction. Let $H'\defeq H+2\cdot\gamma$.

    \begin{itemize}
    \item $\zeta \leq H'$: We choose $\vec{c} \defeq \vec{z} + \zeta
      \cdot \vec{v}$, $\vec{x} = \vec{y} \defeq {\vec{u}}$ and observe that
$\norm{\vec{c}}\leq\norm{P}^{O(1)}$.

    \item $\zeta > H'$: Since $\vec{r} = \vec{z} + \lambda \cdot
      \vec{u} + \zeta \cdot \vec{v}$,  it suffices to show that $\lambda \cdot \vec{u} +
      \zeta \cdot \vec{v}$ can be written as $\varrho \cdot \vec{u} +
      \psi \cdot (0,\alpha) + \omega \cdot \vec{v}$ with $\omega \leq
      \norm{P}^{O(1)}$. To this end, we first rewrite $\lambda \cdot
      \vec{u} + \zeta \cdot \vec{v}$ as
      \begin{eqnarray}
        \lambda \cdot \vec{u} + (\zeta-H'+H') \cdot
        \vec{v} & = & \lambda \cdot \vec{u} +
        \left\lfloor\frac{\zeta-H'}{\gamma}\right\rfloor \cdot
        \gamma \cdot \vec{v} + \underbrace{(((\zeta-H') \bmod
          \gamma)+H')}_{\omega} \cdot \vec{v}\nonumber \\
      & \stackrel{(\ref{E Up})}{=} &
        \underbrace{\left(\lambda -
          \left\lfloor\frac{\zeta-H'}{\gamma}\right\rfloor
          \cdot \beta\right)}_{\varrho} \cdot \vec{u} +
        \underbrace{\left(\left\lfloor\frac{\zeta-H'}{\gamma}\right\rfloor\right)}_{\psi}
        \cdot (0,\alpha) + \omega \cdot \vec{v}.\nonumber
      \end{eqnarray}
\noindent
      Since $\psi \geq 0$ and $\omega \leq (B^4+H') \leq
      \norm{P}^{O(1)}$, it remains to prove that $\varrho \geq 0$:
      \begin{eqnarray*}
        \varrho\ =\  
        \lambda -
        \left\lfloor\frac{\zeta-H'}{\gamma}\right\rfloor \cdot
        \beta & =& 
\lambda- \left\lfloor\frac{\zeta-H-2\cdot\gamma}{\gamma}\right\rfloor \cdot
        \beta 
\\
&>& \lambda -
        \frac{\zeta-H-\gamma}{\gamma} \cdot \beta\nonumber\\
      &  \stackrel{(\ref{E Lambda Ineq})}{>}&\frac{(\zeta-H) \cdot
          \beta}{\gamma} - \beta
        -\frac{\zeta-H-\gamma}{\gamma} \cdot \beta\nonumber\\
& = & 0.\nonumber
      \end{eqnarray*}
      Consequently, we set $\vec{c} = \vec{z} + \omega \cdot \vec{v}$,
      $\vec{x} = \vec{u}$ and $\vec{y} = (0, \alpha)$.
    \end{itemize}

    {\em Case 2: $\vec{u}, \vec{v} \not\in \N^2$}. The case
    $\cone_{\Q}(\{\vec{u}, \vec{v}\}) \, \cap \, \N^2= \{ \vec{0}\}$
    is trivial. Hence we assume $\cone_{\Q}(\{\vec{u}, \vec{v}\}) \,
    \cap \, \N^2 \neq \{ \vec{0}\}$ and it is easily seen that this
    implies $\N^2\subseteq \cone_{\Q}(\{\vec{u}, \vec{v}\})$.  Without
    loss of generality we assume that $u_1,v_2 < 0$ and $u_2,v_1 > 0$
    and consequently have $u_1/u_2 > v_1 / v_2$,
    Figure~\ref{fig:case2} illustrates this case. In particular, for
    all $\lambda',\zeta'\in \Z$ we have
    \begin{align}
      \label{eqn:always-non-negative}
      \lambda'\cdot \vec{u} + \zeta'\cdot \vec{v} \in \N^2\qquad \text{ implies
      } \qquad \lambda',\zeta'\in \N\quad .
    \end{align}
\noindent
    Analogously to Case 1 there exist
    $\sigma,\tau,\xi,\alpha,\beta,\gamma \in [1,B^4]$ with
    \begin{eqnarray}
      \label{eqn:scaled-unit-vectors}
      (\sigma,0) = \tau\cdot\vec{u} + \xi\cdot\vec{v} \in L(\vec{b};
      P) \qquad \text{and} \qquad (0,\alpha) = \beta\cdot\vec{u} +
      \gamma\cdot\vec{v} \in L(\vec{b}; P)\quad .\label{xy axis}
    \end{eqnarray}

    Similar to Case 1, it is sufficient to rewrite $(\ell_1,\ell_2)
    \defeq \lambda \cdot \vec{u} + \zeta \cdot \vec{v}$ as $\varrho
    \cdot (\sigma,0) + \psi \cdot (0,\alpha) + \vec{w}'$, where
    $\varrho, \psi \in \N$, $\vec{w}' \in \cone_\N\{\vec{u},
    \vec{v}\}$ and $\norm{\vec{w}'} \leq \norm{P}^{O(1)}$. We observe
    that $\min(\ell_1,\ell_2)\ge -\norm{\vec{z}}$ and
    $\max(\ell_1,\ell_2) \ge 0$ and make a case distinction.
\medskip

\noindent
    {\em Case 2(a): $\ell_1, \ell_2 \geq 0$.} We have $(\ell_1,\ell_2)
    = (h_1+r_1,h_2+r_2)$, where $h_1 =
    \left\lfloor\frac{\ell_1}{\sigma}\right\rfloor \cdot \sigma$, $h_2
    = \left\lfloor\frac{\ell_2}{\alpha}\right\rfloor \cdot \alpha$,
    $r_1 = (\ell_1\text{ mod }\sigma)$ and $r_2 = (\ell_2 \text{ mod
    }\alpha)$. Due to
    \begin{eqnarray*}
      (h_1,h_2) =
      \left\lfloor\frac{\ell_1}{\sigma}\right\rfloor\cdot(\sigma,0) +
      \left\lfloor\frac{\ell_2}{\alpha}\right\rfloor\cdot(0,\alpha)
      \stackrel{(\ref{xy axis})}{=} \underbrace{ \left(
        \left\lfloor\frac{\ell_1}{\sigma}\right\rfloor \cdot \tau +
        \left\lfloor\frac{\ell_2}{\alpha}\right\rfloor \cdot \beta
        \right)}_{\theta} \cdot \vec{u} + \underbrace{ \left(
        \left\lfloor \frac{\ell_1}{\sigma} \right\rfloor \cdot \xi +
        \left\lfloor \frac{\ell_2}{\alpha} \right\rfloor \cdot
        \gamma\right)}_{\mu}\cdot\vec{v},
    \end{eqnarray*}
    we set $\varrho \defeq
    \left\lfloor\frac{\ell_1}{\sigma}\right\rfloor$ and $\psi \defeq
    \left\lfloor\frac{\ell_2}{\alpha}\right\rfloor$. We argue that we
    can take $\vec{w}' \defeq (r_1,r_2)\in\N^2$. Since
    $\norm{\vec{w}'} \leq B^4$ it remains to show that $\vec{w} \in
    \cone_\N(\{\vec{u}, \vec{v}\})$. To see the latter, we have $(r_1,
    r_2) = (\ell_1, \ell_2) - (h_1, h_2) = (\lambda - \theta)\cdot
    \vec{u} + (\zeta - \mu) \cdot \vec{v}$, and $\lambda - \theta,
    \zeta - \mu \in \Z$. But then (\ref{eqn:always-non-negative})
    yields $\lambda - \theta, \zeta - \mu \in \N$, as required.
\medskip

\noindent
    {\em Case 2(b): $\ell_1<0,\ell_2\geq 0$.} First, we show that in
    all linear combinations $\ell_1 = \lambda \cdot u_1 + \zeta \cdot
    v_1$ such that $-\norm{\vec{z}}\le \ell_1 < 0$, $\lambda$ and
    $\zeta$ only differ by a linear factor. Indeed, we have
    \begin{align}
      \label{eqn:linear-factor}
      \ell_1 = \lambda \cdot u_1 + \zeta \cdot v_1\quad \iff\quad
      \lambda = -\frac{v_1}{u_1} \cdot \zeta + \frac{\ell_1}{u_1} \quad\iff\quad
      \zeta  = -\frac{u_1}{v_1}\cdot \lambda + \frac{\ell_1}{v_1}\quad.
    \end{align}
By subtracting and adding $\alpha\cdot k\cdot v_1\cdot u_1$, we get
    \begin{align}
      \label{eqn:x-constant}
   \forall k\in\N:\quad  \ell_1 & = (\lambda - \alpha \cdot k\cdot v_1) \cdot u_1 + (\zeta + 
      \alpha \cdot k\cdot u_1)\cdot v_1.
    \end{align}
    On the other hand, for $k> 0$ we have
    \begin{align}
      \label{eqn:y-drop}
      (\lambda - \alpha \cdot k \cdot v_1) \cdot u_2 + 
      (\zeta + \alpha \cdot k \cdot u_1 )\cdot v_2
      & = \lambda\cdot u_2 + \zeta \cdot v_2 + \alpha \cdot k \cdot 
      (u_1\cdot v_2 - v_1\cdot u_2) \\ & \notag < 
      \lambda \cdot u_2 + \zeta \cdot v_2
    \end{align}
    where the latter inequality follows from $u_1/u_2 > v_1 / v_2$.
    Let us define
$$
k_0\quad\defeq\quad\min\left\{\ \max\{k\in\N: \lambda -\alpha \cdot k\cdot v_1\geq 0\},\ \max\{k\in\N:
\zeta + \alpha \cdot k \cdot u_1\ge 0\}\ \right\}\quad.
$$
Moreover, let $\lambda'\defeq \lambda -\alpha \cdot k_0
    \cdot v_1\geq 0$ and $\zeta'\defeq \zeta + \alpha \cdot k_0 \cdot
    u_1\ge 0$. 
Clearly $\min\{\lambda',\zeta'\}\leq\norm{P}^{O(1)}$ by the choice of $k_0$
and hence $\lambda',\zeta'\leq\norm{P}^{O(1)}$ by (\ref{eqn:linear-factor}).
Moreover, from (\ref{eqn:x-constant}) we have that $\ell_1 =
    \lambda' \cdot u_1 + \zeta'\cdot v_1$.
 Finally, as required, we have
    \begin{eqnarray*}
      \begin{vmatrix} \ell_1\\ \ell_2 \end{vmatrix}
        \quad  &=&\lambda\cdot\vec{u}+\zeta\cdot\vec{v}\\
&\stackrel{(\ref{eqn:y-drop})}{=}&
      \lambda' \cdot \vec{u} + \zeta' \cdot
      \vec{v} + \underbrace{(- k_0\cdot (u_1\cdot v_2 - v_1 \cdot u_2))}_{\geq 0} \cdot
      \begin{vmatrix} 0\\ \alpha \end{vmatrix}\quad .
    \end{eqnarray*}
\medskip
\noindent
    {\em Case 2(c): $\ell_1\geq0,\ell_2<0$.} This case is symmetric to
    Case 2(b) and therefore omitted.

    \begin{figure}[t]
      \begin{center}
        \includegraphics{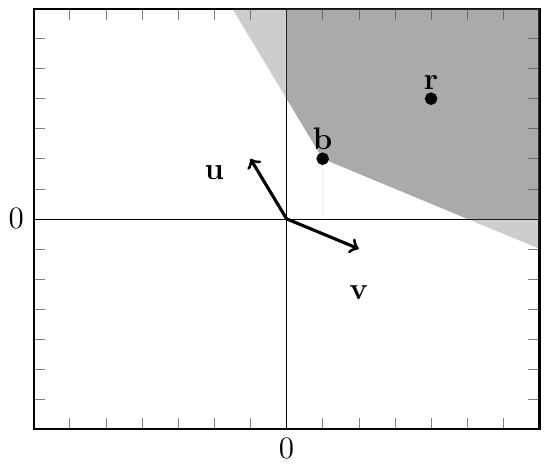}

            



            
      \end{center}

      \caption{Example of Case 2. The filled area corresponds to
        $\vec{b} + \cone_{\Q_{\geq 0}}(\{\vec{u}, \vec{v}\})$
        and the darker filled area corresponds to $\vec{b} +
        \cone_{\Q}(\{\vec{u}, \vec{v}\}) \, \cap\ (\Q_{\geq 0} \times
        \Q_{\geq 0})$. }
      \label{fig:case2}
    \end{figure}
\end{proof}

Let us give an intuitive idea of how we can prove Proposition~\ref{P
  upper right}~(a) by an application of Lemma~\ref{L Dimension
  Two}. Suppose we are given a run starting in $q(u_1,u_2)$ and ending
in $q(v_1,v_2)$ such that $u_1\le v_1$ and $u_2\le v_2$. From
Proposition~\ref{P zreach} we know that the $\Z$-reachability relation
can be captured by a union of linear path schemes. Since we start and
end in the \emph{same} state, any such linear path scheme can
equivalently be viewed as a linear set $L(\vec{b};P)$ such that
$\vec{b}\in P$. An application of Lemma~\ref{L Dimension Two} then
allows us to decompose such a linear set into a semi-linear set whose
period vectors all point into the same $\N^2$ direction. The crucial
point is that any linear set in this semi-linear set can again be
translated back into a linear path scheme with at most two cycles
whose displacements point to $\N^2$. Consequently, any path obtained
from such a linear path scheme does not, informally speaking, drift
away too much, and if $u_1$ and $u_2$ are sufficiently large then
$\N$-reachability and $\Z$-reachability coincide.

In order to make our intuition formal, we introduce some further
additional notation. Interpreting Lemma~\ref{L
  Dimension Two} in terms of linear path schemes allows us to
establish the following lemma.

\begin{lemma}
  \label{lem:zig-zag-free-lps}
  Let $q \in Q$. For every linear path scheme $\rho$ from $q$ to $q$,
  there exists a finite set $R_\rho$ of zigzag-free linear path
  schemes such that
  \begin{enumerate}[(i)]
  \item $\delta(\rho) \subseteq \delta(R_\rho)$,
  \item $|\sigma| \leq (|\rho| + \norm{T})^{O(1)}$ for each
    $\sigma \in R_\rho$, and
  \item each $\sigma \in R_\rho$ has at most two cycles.
  \end{enumerate}
\end{lemma}

\begin{proof}
  Let $\rho = \alpha_0 \beta_1^* \alpha_1 \cdots \beta_k^* \alpha_k$
  be a linear path scheme. Without loss of generality we assume that
  $\delta(\alpha_0 \cdots \alpha_k) \in \{\delta(\beta_i) :
  i\in[1,k]\}$, otherwise we apply the claim to the linear path scheme
  $\rho' \defeq (\alpha_0 \cdots \alpha_k)^* \rho$ which satisfies
  this property and for which we have $\delta(\rho) \subseteq
  \delta(\rho')$. Moreover we assume that $\delta(\beta_i) \not=
  \delta(\beta_j)$ whenever $i \not= j$, since otherwise we can just
  remove $\beta_j$ from $\rho$ which results in a linear path scheme
  with the same displacements as $\rho$. We can
  write $$\delta(\rho)\quad =\quad \bigcup_{Z\ \text{is a quadrant}}
  \delta(\rho) \cap Z\quad.$$ Hence for each quadrant $Z$ it is
  sufficient to construct a set of appropriate zigzag-free linear path
  scheme $R_{\rho,Z}$ such that $\delta(\rho)\cap Z =
  \delta(R_{\rho,Z})$ since we can then just define our set of linear
  path schemes as $R_{\rho} \defeq \bigcup\{R_{\rho,Z}: Z\text{ is a
    quadrant}\}$.  Let $\vec{b} \defeq
  \delta(\alpha_0 \cdots \alpha_k)$ and let $P \defeq
  \{\delta(\beta_i) : i \in [1,k]\}$. By assumption we have $\vec{b}
  \in P$. Note that $\norm{P} \leq \abs{\rho} \cdot \norm{T}$. By
  Lemma~\ref{L Dimension Two} there exists a semi-linear set
  $\bigcup_{i \in I} L(\vec{c}_i; P_i)$ with $\delta(\vec{b}; P) \cap
  Z = \bigcup_{i \in I} L(\vec{c}_i; P_i)$ satisfying for each $i
  \in I$,
  \begin{itemize}
  \item $|P_i| \leq 2$,
  \item $P_i \subseteq (P \cup L(\vec{b}; P)) \cap Z$,
    and
  \item there exists $e\leq\norm{P}^{O(1)}$ such that
    $\{\vec{c}_i\}\cup(P_i\cap
    L(\vec{b};P))\subseteq\vec{b}+\cone_{[0,e]}(P)$.
  \end{itemize}
  Let us fix an arbitrary $i\in I$. By the last item for each $\vec{u} \in
  \{\vec{c}_i\}\cup (P_i \cap L(\vec{b}; P))$ there exists a path
  $\pi_{\vec{u}}$ from $q$ to $q$ of the form $\alpha_0 \beta_1^{e_1}
  \alpha_1 \cdots \beta_k^{e_k} \alpha_k$ for some $0 \leq e_1,
  \ldots, e_k \leq \norm{P}^{O(1)} \leq (|\rho| + \norm{T})^{O(1)}$ 
  with $\vec{u} = \delta(\pi_{\vec{u}})$; thus $|\pi_{\vec{u}}| \leq
  (|\rho| + \norm{T})^{O(1)}$.  Let $\pi_{\vec{c}_i} = \alpha_0
  \beta_1^{e_1} \alpha_1 \cdots \beta_k^{e_k} \alpha_k$ and define the
  linear path scheme $\sigma_i$ to be obtained from $\pi_{\vec{c}_i}$ by inserting
  appropriate cycles $\beta_j^*$ whenever $\delta(\beta_j) \in P_i
  \cap P$. Formally, we define
  \begin{align*}
    \sigma_i \defeq \alpha_0 \beta_1^{e_1} \theta_1 \alpha_1 \cdots
    \beta_k^{e_k} \theta_k \alpha_k, \quad \text{where } \theta_j =
    \begin{cases}
      \beta_j^* & \text{ if } \delta(\beta_j) \in P_i \cap P
      \\ \varepsilon & \text{otherwise}
    \end{cases}
    \quad \text{for every } j \in [1,k]\quad.
  \end{align*}
  Recalling that $P_i\subseteq (P \cup L(\vec{b}; P)) \cap Z$, it is now readily seen that
  \begin{align*}
    \rho_i\ \defeq\ \sigma_i \cdot \prod_{\vec{u} \in P_i \setminus P} \pi_{\vec{u}}^*
  \end{align*}
  is a zigzag-free linear path scheme with at most two cycles whose
  displacements point to $Z$, and satisfying $\delta(\rho_i) =
  L(\vec{c}_i; P_i)$ and $|\rho_i| \leq (|\rho| +
  \norm{T})^{O(1)}$. Finally, we define $R_{\rho,Z} \defeq
  \bigcup_{i \in I} \rho_i$ due to
  \begin{align*}
    \delta(\rho) \cap Z\ =\ L(\vec{b}; P) \cap Z\ =\ \bigcup_{i \in I} L(\vec{c}_i; P_i)\ =\ \bigcup_{i
      \in I} \delta(\rho_i) = \delta(R_{\rho,Z}).
  \end{align*}
\end{proof}

\noindent
We are now fully prepared to give a proof of Proposition~\ref{P upper
  right}.

\begin{proof}[Proof of Proposition~\ref{P upper right}]
Let us fix a 2-VASS $V=(Q,T)$.
\medskip

\noindent
{\em Proof of (a): }
Let $S$ be the finite set of linear path scheme from Proposition~\ref{P zreach}
    such that
    \begin{itemize}
    \item $p(\vec{u}) \xrightarrow{*}_{\Z^d} q(\vec{v})$ if, and only
      if, $p(\vec{u}) \xrightarrow{S}_{\Z^d} q(\vec{v})$,
    \item $|\rho|\leq 2\cdot |Q|\cdot|T|$ for each $\rho \in S$ and
    \item each $\rho\in S$ has at most $|T|$ cycles.
    \end{itemize}
    We apply Lemma~\ref{lem:zig-zag-free-lps} to each $\rho \in S$ and
    define $R \defeq \bigcup_{\rho \in S} R_\rho $. Hence, for each
    $\sigma \in R$ we have $|\sigma| \leq (|T| \cdot |Q| +
    \norm{T})^{O(1)} = (|Q| + \norm{T})^{O(1)}$ by (ii) of
    Lemma~\ref{lem:zig-zag-free-lps}. We set $D$ required in
    Proposition~\ref{P upper right} to $D \defeq \max\{|\sigma| :
    \sigma \in R\} \cdot \norm{T} \leq (|Q| + \norm{T})^{O(1)}$.  The
    monotonicity of zigzag-free linear path schemes now provides the
    key ingredient for proving Proposition~\ref{P upper
      right}~(a). For the rest of the proof let us fix
    $\vec{u},\vec{v}\in[D,\infty)^2$ and some zigzag-free linear path
      scheme $\sigma = \alpha_0 \beta_1^* \alpha_1 \beta_2^* \alpha_2
      \in R$. Suppose $q(\vec{u}) \xrightarrow{\pi}_{\Z^2} q(\vec{v})$
      for some $\pi = \alpha_0 \beta_1^{e_1} \alpha_1 \beta_2^{e_2}
      \alpha_2$, then by definition of $D$ it is clear that
    \begin{align}
      \label{eqn:zig-zag-drop}
      \forall i\in[0,|\pi|]:\qquad \vec{0}\quad \le\quad \vec{u} + \delta(\pi[1,i]) \quad\le\quad \vec{v}
      + \begin{pmatrix} D\\ D
      \end{pmatrix}\quad.
    \end{align} 
    It remains to prove $q(\vec{u}) \xrightarrow{*}_{\N^2} q(\vec{v})$
    if, and only if, $q(\vec{u}) \xrightarrow{\sigma}_{\N^2}
    q(\vec{v})$ for some $\sigma \in R$. The latter follows from the
    following circular sequence of implications and equivalences:
    \begin{eqnarray*}
      q(\vec{u}) \xrightarrow{*}_{\N^2} q(\vec{v}) & \quad
      \Longrightarrow \quad & q(\vec{u}) \xrightarrow{*}_{\Z^2}
      q(\vec{v}) \\
      & \stackrel{\text{Proposition~\ref{P zreach}}}{\Longleftrightarrow} &
      q(\vec{u}) \xrightarrow{\rho}_{\Z^2} q(\vec{v}) \text{ for some
        $\rho \in S$} \\
      & \stackrel{\text{Lemma~\ref{lem:zig-zag-free-lps}(i)}}{\Longrightarrow} & q(\vec{u})
      \xrightarrow{\sigma}_{\Z^2} q(\vec{v}) \text{ for some $\sigma \in
        R_\rho$ for some $\rho \in S$} \\
      & \stackrel{(\ref{eqn:zig-zag-drop})}
                  {\Longrightarrow} & q(\vec{u}) \xrightarrow{\sigma}_{\N^2}
                  q(\vec{v}) \text{ for some $\sigma \in R$}\\
                  & \Longrightarrow & q(\vec{u}) \xrightarrow{*}_{\N^2} q(\vec{v})
    \end{eqnarray*}
\ \\ 
\medskip
\noindent
{\em Proof of (b): }
 Suppose that $p(\vec{u}) \xrightarrow{\pi}_\O q(\vec{v})$.
    Then $\pi$ can be factorized as $\pi = \alpha_0 \beta_1 \alpha_1
    \cdots \beta_k \alpha_k$ such that
    $$ p(\vec{u}) \xrightarrow{\alpha_0}_\O q_1(\vec{u}_1)
    \xrightarrow{\beta_1}_\O q_1(\vec{u}_1') \xrightarrow{\alpha_1}_\O
    q_2(\vec{u}_2)\ \cdots\ q_k(\vec{u}_k) \xrightarrow{\beta_k}_\O
    q_k(\vec{u}_k') \xrightarrow{\alpha_k}_\O q(\vec{v})
    $$ 
    where $|\alpha_0|, |\alpha_1|, \dots, |\alpha_k| \leq |Q|$, each
    $\beta_i$ is a cycle from $q_i$ to $q_i$ for some $q_i \in Q$, and
    $k \leq |Q|$. Since $\vec{u}_i, \vec{u}_i'\in \O$ for all $i \in
    [1,k]$, by (a) we have $ q_i(\vec{u}_i)
    \xrightarrow{\rho_i}_{\N^2} q_i(\vec{u}_i') $ for some linear path
    scheme $\rho_i \in R$. Consequently, we define $X$ as
    \begin{multline*}
      X \defeq \{\alpha_0 \rho_1 \alpha_1 \cdots \rho_k \alpha_k\text{ linear path scheme} :  
k \leq |Q|,  
\alpha_i\in T^*, |\alpha_i| \leq |Q|,  \rho_i \in R)\}
    \end{multline*}
    Let $\rho \in X$, then we have $|\rho| \leq |Q|^2 + |Q| \cdot (|Q| + \norm{T})^{O(1)} = (|Q|
    + \norm{T})^{O(1)}$, and $\rho$ has at most $2\cdot|Q|$
    cycles.

\end{proof}

\subsection{Reachability in 2-VASS with One Bounded Component}
\label{sec:bounded-component}

 The purpose of this section is to establish the following result on
reachability between configurations for which 
there exists a run 
on which for all configurations
 {\em at most one} of the two counter
values exceeds a certain bound.
 We refer to the
bottom picture of Figure \ref{fig:types:paths}.

\begin{proposition}{\label{P Belts}}
  Let $V = (Q,T)$ be a 2-VASS, $D \in \N$ and $\L = ([0,D] \times \N)
  \cup (\N \times [0,D])$. There exists a finite set $Y_\L$ of linear
  path schemes such that
  \begin{itemize}
  \item $p(\vec{u}) \xrightarrow{*}_\L q(\vec{v})$ implies $p(\vec{u})
    \xrightarrow{Y_\L}_{\N^2} q(\vec{v})$,
  \item $|\rho| \leq (|Q| + \norm{T} + D)^{O(1)}$ for every $\rho \in
    Y_\L$; and
  \item each $\rho \in Y_\L$ has at most two cycles.
  \end{itemize}
\end{proposition}
In its essence, restricting the set of admissible values of one of the two counters
of a 2-VASS to $[0,D)$ as in Proposition~\ref{P Belts} gives rise to a
  1-VASS.
  This observation enables us to resort to techniques and
  results developed for 1-VASS (in fact even for one-counter automata),
  respectively. In particular, subsequently we make use of the
  following lemma established by Valiant and Paterson.  It shows that
  reachability in a 1-VASS is captured by a finite union of linear path schemes each having
  at most one cycle.

\begin{lemma}[Lemma 2 in~\cite{VP75}]\label{L VP75}
  Let $V=(Q,T)$ be a 1-VASS with unary updates (i.e. $T\subseteq Q\times\{-1,0,1\}\times Q$) and let $p(u)
  \xrightarrow{*}_\N q(v)$ for some configurations $p(u)$ and $q(u)$
  such that $|u-v| \geq |Q|+|Q|^2$. There exist $\alpha, \beta, \gamma
  \in T^*$ and $\pi \in T^*$ such that $p(u) \xrightarrow{\pi}_{\N} p(v)$ and $\pi$ has the the
  following properties,
  \begin{itemize}
  \item $\pi = \alpha\beta^i\gamma$ for some $i>0$,
  \item $\alpha\beta^*\gamma$ is a linear path scheme with one cycle, and
  \item $|\alpha\gamma| < |Q|^2$ and $\beta$ is a cycle with $|\beta|
    \leq |Q|$ and $|\delta(\beta)| \in [1,|Q|]$.
  \end{itemize}
\end{lemma}

The following lemma states that in a 1-VASS with unary updates
 between any two reachable configurations
with absolute counter difference $D$ there is a run witnessing their
reachability that has length at most $|Q|^{O(1)}+|Q|\cdot D$. It is
obtained as an easy consequence of Lemma~\ref{L VP75}.

\begin{lemma}\label{prop:1vass:flat:pmone}
  Let $V=(Q,T)$ be a 1-VASS with unary updates, i.e. $T\subseteq
  Q\times\{-1,0,1\}\times Q$.  Let $u,v\in \N$ and $D\defeq
  \abs{v-u}$. If $p(u) \xrightarrow{*}_{\N} q(v)$ then there is some
  run $p(u) \xrightarrow{\pi}_{\N} q(v)$ with $|\pi| \leq
  |Q|^{O(1)}+|Q|\cdot D$.
\end{lemma}

\begin{proof}
  We first consider the case when $D \geq |Q|+|Q|^2$. By Lemma~\ref{L
    VP75}, we have $p(u) \xrightarrow{\alpha\beta^i\gamma}_{\N} q(v)$
  for some $i\ge 0$, where $\alpha\beta^*\gamma$ is a linear path scheme,
  $|\alpha\gamma| < |Q|^2$ and $\beta$ is a cycle with $|\beta| \leq
  |Q|$ and $|\delta(\beta)| \in [1,|Q|]$. Since $|\delta(\beta)| \in
  [1,|Q|]$ we have $i \leq D+|\alpha\gamma|$ and hence
  $|\alpha\beta^i\gamma| < |Q|^2+|Q|\cdot i\leq
  |Q|^{O(1)}+|Q|\cdot D$.

  We now turn to the case in which $D=|u-v|< |Q|+|Q|^2$. By the pigeonhole
  principle, for any run from $p(u)$ to $q(v)$ of minimal length either
  \begin{enumerate}[(i)]
  \item every configuration $r(w)$ on this minimal run satisfies
    $|w-v| < 2\cdot(|Q|+|Q|^2)$, or 
  \item there exists an intermediate configuration $r(w)$ on this minimal run with $|w-v|
    = 2\cdot(|Q|+|Q|^2)$.
    \end{enumerate}
  Clearly, any run of the form (i) is of length strictly less than
  $4\cdot(|Q|+|Q|^2)\cdot|Q|$. Otherwise, for any minimal run $\pi$
  from $p(u)$ to $q(v)$ of the form (ii) there is some configuration
  $r(w)$ along this path with $|w-v| = 2\cdot(|Q|+|Q|^2)$. Note that
  we have $|w-v| \leq |w-u| + |u-v|$ by the triangle inequality. This
  allows to conclude $|Q|+|Q|^2 \leq |w-u| \leq 3\cdot(|Q|+|Q|^2)$ due
  to
  \begin{align*}
    |Q|+|Q|^2\ \leq\  |w-v|-|u-v|\ \leq\ |w-u|\ \leq\ |w-v|+|u-v|\ \leq\
    3\cdot(|Q|+|Q|^2).
  \end{align*} 
  Summarizing, we have $|Q|+|Q|^2\leq |w-u|\leq 3\cdot (|Q|+|Q|^2)$ and 
$|w-v|=2\cdot(|Q|+|Q|^2)$.
Thus, $|\pi|$ is at most the length of two runs each of which has a
counter difference of at least $|Q|+|Q|^2$, namely 
the length of a minimal run from $p(u)$ to $r(w)$ plus the
 length of a minimal run from $r(w)$ to $q(v)$:
  \begin{align*}
    |\pi|\quad \leq\quad (|Q|^2+|Q| \cdot 3\cdot (|Q|+|Q|^2)) + (|Q|^2 +
    |Q|\cdot 2\cdot(|Q|+|Q|^2)) \quad\leq\quad |Q|^{O(1)}
  \end{align*}
\end{proof}

We now combine the Lemmas~\ref{L VP75} and~\ref{prop:1vass:flat:pmone}
in order to show that the reachability relation of a 1-VASS $V=(Q,T)$
(with binary updates) can be captured by a union of linear path
schemes that each have at most one cycle and length polynomially
bounded in $|Q|+\norm{T}$.

\begin{lemma}\label{prop:1vass:flat}
  Let $V=(Q,T)$ be a 1-VASS. There exists a finite set $Y$ of linear path schemes such
  that
  \begin{enumerate}[(i)]
  \item $p(u) \xrightarrow{*}_{\N} q(v)$ if, and only if, $p(u)
    \xrightarrow{Y}_{\N} q(v)$,
  \item $|\rho| \leq (|Q| + \norm{T})^{O(1)}$ for each $\rho \in Y$, and
  \item each $\rho \in Y$ has at most one cycle.
  \end{enumerate}
\end{lemma}

\begin{proof}
  The idea is to construct from $V=(Q,T)$ a unary 1-VASS $V'=(Q',T')$
  with $T' \subseteq Q' \times \{-1,0,+1\} \times Q'$ that mimics the
  behavior of $V$. We then apply Lemmas~\ref{L VP75}
  and~\ref{prop:1vass:flat:pmone} to $V'$ in order to obtain the set
  of linear path schemes $Y$ for $V$. We mimic every transition $t = (q,z,q') \in T$
  by a sequence of $|z|+2$ transitions in $V'$ of which $|z|$ either
  all increment or decrement the counter. Consequently, we define $$Q'
  \quad\defeq\quad Q \cup \{(t,i) : t = (p,z,q) \in T, i \in [0,|z|]\}$$ and
  \begin{eqnarray*} 
    T' &\ \defeq\ &\phantom{\cup\ } \{(p,0,(t,0)) : t = (p,z,q) \in
    T\} \cup \{((t,|z|),0,q) : t = (p,z,q) \in T\} \\
    && \cup\ \{((t,i),+1,(t,i+1)) : t = (p,z,q) \in T, z>0, i \in
    [0,\abs{z}-1]\} \\
    && \cup\ \{((t,i),-1,(t,i+1)) : t = (p,z,q) \in T, z<0, i \in
    [0,|z|-1]\}\quad.
  \end{eqnarray*} 
  Let us define the homomorphism $h:T\rightarrow T'^+$ such that
  $$ h(t)\ \defeq\
  \begin{cases}
    (q,0,(t,0)) \cdot \left(\prod_{i=1}^{|z|}
    ((t,i-1),+1,(t,i))\right) \cdot ((t,|z|),0,q') & \text{ if $t =
      (q,z,q')$ and $z \geq 0$} \\
    
    (q,0,(t,0)) \cdot \left(\prod_{i=1}^{|z|}
    ((t,i-1),-1,(t,i))\right) \cdot ((t,|z|),0,q') & \text{ if $t =
      (q,z,q')$ and $z < 0$}
  \end{cases}
  $$ for every $t \in T$. The idea behind this definition is that for
  every run $\pi$ in $V$ we have that $h(\pi)$ is the run in $V'$
  that corresponds to $\pi$. The following conditions formalize this
  intuition and are easily verified:
  \begin{enumerate}[(i)]
  \item $|h(t) |\leq \norm{T}+2$ for each $t \in T$,
  \item if $p(u) \xrightarrow{\pi}_{\N} q(v)$ in $V$ then
    $p(u) \xrightarrow{h(\pi)}_{\N} q(v)$ in $V'$, and
  \item if $p,q \in Q$ and $p(u) \xrightarrow{\pi'}_{\N} q(v)$ in $V'$
    then there is a unique $\pi \in T^*$ satisfying $\pi'= h(\pi)$ and
    $p(u) \xrightarrow{\pi}_{\N} q(v)$ in $V$.
  \end{enumerate}
  By (iii) for every $p(u) \xrightarrow{\pi'}_{\N} q(v)$ with $p,q\in Q$ in $V'$ we
  can write $h^{-1}(\pi')$ to denote the unique $\pi$ such that
  $h(\pi) = \pi'$ and $p(u) \xrightarrow{\pi}_{\N} q(v)$ in $V$. In
  this case, we have that $\pi'$ is a cycle in $V'$ if, and only if,
  $h^{-1}(\pi')$ is a cycle in $V$.

  To show the existence of the finite set of linear path schemes satisfying the
  conditions required in the lemma, we show that whenever $p(u)
  \xrightarrow{*}_{\N} q(v)$ in $V$ then there exists a linear path scheme $\rho
  \subseteq T^*$ such that $p(u) \xrightarrow{\rho}_{\N} q(v)$ in $V$
  and $|\rho| \leq (|Q| + \norm{T})^{O(1)}$. Let $D = |Q'|+|Q'|^2$ and
  assume $p(u) \xrightarrow{*}_{\N} q(v)$ in $V$. Hence $p(u)
  \xrightarrow{*}_{\N} q(v)$ in $V'$ by (ii). We make a case
  distinction between $|u-v| \leq D$ and $|u-v| > D$.
  
  {\em Case 1: $|u-v|\leq D$.} By Lemma~\ref{prop:1vass:flat:pmone} we
  have $p(u) \xrightarrow{\pi'}_{\N} q(v)$ in $V'$ for some path
  $\pi'$ with $|\pi'| \leq |Q'|^{O(1)} + |Q'|\cdot D \leq (|Q| + \norm{T} +
  D)^{O(1)}$. Thus, we set $\rho = h^{-1}(\pi')$ and note that $p(u)
  \xrightarrow{\rho}_{\N} q(v)$ in $V$ by (iii), and $|\rho| \leq|
  \pi'| \leq(|Q| + \norm{T} + D)^{O(1)}\leq (|Q|+\norm{T})^{O(1)}$ as
  required.

  {\em Case 2: $|u-v|>D$.}
  By Lemma~\ref{L VP75}, we have $p(u)
  \xrightarrow{\alpha'(\beta')^i\gamma'}_{\N} q(v)$ in $V'$ for some
  $i > 0$ and some linear path scheme $\rho' =\alpha'\beta'^*\gamma'$ from $p$
  to $q$ satisfying $|\alpha'\gamma'| < |Q'|$ and $|\beta'| \leq
  |Q'|$.

  Let $q'\in Q'$ be such that $\beta'$ is a
  cycle from $q'$ to $q'$. If $q'\in Q$ then $\alpha'$ is a path in
  $V'$ from $p$ to $q'$, $\beta'$ is a cycle in $V'$ from $q'$ to $q'$
  and $\gamma'$ is a path in $V'$ from $q'$ to $q$. Thus $\rho \defeq
  \alpha'^{-1}(\beta'^{-1})^*\gamma'^{-1}$ is a linear path scheme in
  $V$ for which we have $p(u) \xrightarrow{\rho}_{\N} q(v)$ in $V$ by
  (iii), and $|\rho| \leq |Q'|+|Q'|^2 \leq (|Q| + \norm{T})^{O(1)}$.

  Otherwise, if $q' \in Q' \setminus Q$, we have $q'=(t,i)$ for some
  $t = (q_1,z,q_2) \in T$ and some $i \in [1,|z|]$. We only consider the
  case $z \geq 0$, the case $z<0$ being symmetric. Since $\beta'$ is a
  cycle from $(t,i)$ to $(t,i)$, it follows from the definition of
  $V'$ that $\alpha' = \alpha_1'\alpha_2'$ and $\beta' =
  \beta_1'\beta_2'$, where $\alpha_2' = \beta_2' = (q_1,0,(t,0)) \cdot
  \prod_{j=1}^i((t,j-1),+1,(t,j))$. We have the following
language equalities, where the last equality follows from
$\alpha_2'=\beta_1'$,
  $$\alpha'(\beta')^*\gamma'\quad=\quad
  \alpha_1'\alpha_2'(\beta_1'\beta_2')^*\gamma'
  \quad
=\quad
  \alpha_1'(\beta_2'\beta_1')^*\beta_2'\gamma'.$$
  Moreover, in $V'$ it holds that $\alpha_1'$ is a path from $p\in Q$
  to $q_1\in Q$, $\beta_2'\beta_1'$ is a cycle from $q_1\in Q$ to
  $q_1$, and $\beta_2'\gamma'$ is a path from $q_1$ to $q\in Q$. Hence
  $\rho \defeq h(\alpha_1')^{-1} (h(\beta_2'\beta_1')^{-1})^*
  h(\beta_2'\gamma')^{-1}$ is a linear path scheme in $V$ with $|\rho|
  \leq |Q'|+|Q'|^2 \leq (|Q| + \norm{T})^{O(1)}$ for which we have
  $p(u) \xrightarrow{\rho}_{\N} q(v)$ in $V$ by (iii).
\end{proof}

We are now in a position where we, informally speaking, can prove the
first half of Proposition~\ref{P Belts}. The following lemma proves
Proposition~\ref{P Belts} when restricting the range of one
counter.

\begin{lemma}{\label{C Belt1}}
  Let $V = (Q,T)$ be a 2-VASS, $D \in \N$ and $\B \in \{(\N \times
  [0,D]), ([0,D]\times \N)\}$. Then there exists a finite set $Y_\B$ of
  linear path schemes such that
  \begin{itemize}
  \item $p(\vec{u}) \xrightarrow{*}_\B q(\vec{v})$ if, and only if,
    $p(\vec{u}) \xrightarrow{Y_\B}_\B q(\vec{v})$;
  \item $|\rho| \leq (|Q| + \norm{T} + D)^{O(1)}$ for each $\rho
    \in Y_\B$; and
  \item each $\rho \in Y_\B$ has at most one cycle.
  \end{itemize}
\end{lemma}

\begin{proof}
  We only consider the case $\B=\N \times [0,D]$, the other case
  follows by symmetry. Starting from $V$ we construct a 1-VASS
  $\overline{V} = (\overline{Q},\overline{T})$  such that the
  following holds:
  \begin{enumerate}[(1)]
  \item $\overline{Q} = \{q_i : q \in Q, i \in [0,D]\}$, and
  \item for each $p,q \in Q$ and each $(u_1, u_2), (v_1, v_2) \in \B$
    we have $p(u_1, u_2) \xrightarrow{*}_\B q(v_1, v_2)$ in $V$ if,
    and only if, $p_{u_2}(u_1) \xrightarrow{*}_{\N} q_{v_2}(v_1)$ in
    $\overline{V}$.
  \end{enumerate}
  To achieve (2) note that we can simply define $\overline{T}$ as
  follows,
  $$\overline{T} = \{(p_n,i,q_{n+j}) : (p,(i,j),q) \in T \text{ and } n,n+j
  \in [0,D]\}.$$ 
  This gives rise to a homomorphism $\phi : \overline{T}^* \rightarrow
  T^*$ with $\phi(p_n,i,q_{n+j}) \defeq (p,(i,j),q)$ for each
  $(p_n,i,q_{n+j}) \in \overline{T}$.  For each path $\overline{\pi}$
  in $\overline{V}$ we have
  \begin{itemize}
  \item[(3)] if $p_{u_2}(u_1) \xrightarrow{\overline{\pi}}_\N
    q_{v_2}(v_1)$ in $\overline{V}$ then $p(u_1,u_2)
    \xrightarrow{\phi(\overline{\pi})}_\B q(v_1,v_2)$ in $V$.
  \end{itemize} 

  It follows immediately from the definition of $\overline{T}$ that
  any linear path scheme $\overline{\rho} = \alpha_0 \beta_1^*
  \alpha_1 \cdots \beta_k^* \alpha_k$ over the 1-VASS $\overline{V}$
  induces the linear path scheme $\phi(\overline{\rho}) =
  \phi(\alpha_0) \phi(\beta_1)^* \phi(\alpha_1) \cdots \phi(\beta_k)^*
  \phi(\alpha_k)$ over the 2-VASS $V$. Furthermore, $\phi$ is
  naturally extended to any set of linear path schemes $\overline{S}$:
  we put $\phi(\overline{S}) \defeq \bigcup\{\phi(\overline{\rho}) :
  \overline{\rho} \in \overline{S}\}$. Applying
  Lemma~\ref{prop:1vass:flat} to $\overline{V}$ yields a set
  $\overline{Y}$ of linear path schemes such that
  \begin{itemize}
  \item[(4)] $p(u) \xrightarrow{*}_{\N} q(v)$ in $\overline{V}$ if,
    and only if, $p(u) \xrightarrow{\overline{\rho}}_{\N} q(v)$ in
    $\overline{V}$ for some $\overline{\rho} \in \overline{Y}$, where
    $\overline{\rho}$ has at most one cycle and $|\overline{\rho}|
    \leq (|\overline{Q}| + \norm{\overline{T}})^{O(1)} \leq (|Q|\cdot
    D + \norm{T})^{O(1)} = (|Q| + \norm{T} + D)^{O(1)}$.
  \end{itemize}

  We define $Y_\B$ required in the lemma as $Y_\B \defeq
  \phi(\overline{Y})$. By definition, $Y_\B$ already fulfills the second
  and third condition required in the lemma. The first condition now
  follows from the following circular sequence of implications. Let
  $p,q \in Q$ and $u_1,u_2,v_1,v_2\in\N$, we have
  \begin{eqnarray*}
    p(u_1,u_2) \xrightarrow{*}_\B q(v_1,v_2) \text{ in $V$} & \qquad
    \stackrel{\text{(2)}}{\Longleftrightarrow} \qquad & p_{u_2}(u_1)
    \xrightarrow{*}_{\N} q_{v_2}(v_1) \text{ in $\overline{V}$} \\
    & \qquad \stackrel{\text{(4)}}{\Longleftrightarrow} \qquad & p_{u_2}(u_1)
    \xrightarrow{\overline{Y}}_{\N} q_{v_2}(v_1) \text{ in $\overline{V}$} \\
    & \qquad \stackrel{\text{(3)}}{\Longrightarrow} \qquad & p(u_1,u_2)
    \xrightarrow{\phi(\overline{Y})}_\B q(v_1,v_2) \text{ in $V$} \\
    & \qquad \stackrel{}{\Longleftrightarrow} \qquad & p(u_1,u_2)
    \xrightarrow{S}_\B q(v_1,v_2) \text{ in $V$} \\
    &\Longrightarrow & p(u_1,u_2) \xrightarrow{*}_\B q(v_1,v_2) \text{
      in $V$}.
  \end{eqnarray*}
\end{proof}

In the remainder of this section, by application of Lemma~\ref{C
  Belt1} we prove Proposition~\ref{P Belts} which, given some $D\in
\N$, states that runs which stay inside the $L$-shaped band $\L=([0,D]
\times \N) \cup (\N \times [0,D])$ can be captured by a union of small
linear path schemes with at most two cycles.

\medskip
\noindent
{\em Proof of Proposition \ref{P Belts}. } Let us define $E \defeq D +
\norm{T}$ and $\L' \defeq ([0,E] \times \N) \cup (\N \times
     [0,E])$. Let $p,q \in Q$ and $\vec{u}, \vec{v} \in\L$, and let
     $p(\vec{u}) \xrightarrow{\pi}_\L q(\vec{v})$ be such that $|\pi|$
     is minimal. In order to prove Proposition~\ref{P Belts}, it
     suffices to provide some linear path scheme $\rho$ such that
     $p(\vec{u}) \xrightarrow{\rho}_{\L'} q(\vec{v})$, $|\rho|
     \leq(|Q| + \norm{T} + D)^{O(1)}$ and $\rho$ has at most two
     cycles. Let $\B_1 \defeq[0,E] \times \N$, $\B_2 \defeq \N \times
     [0,E]$ and let $H \defeq \B_1 \cap \B_2 = [0,E] \times
     [0,E]$. Due to minimality of $\pi$ and by choice
     of $H$ we can factorize $\pi$ as $\pi =
     \pi_1\cdots\pi_k$, where
  $$p_0(\vec{u}_0) \xrightarrow{\pi_1}_\L
  p_1(\vec{u}_1) \cdots \xrightarrow{\pi_k}_\L p_k(\vec{u}_k)$$ 
and

\begin{enumerate}[(i)]
  \item $p_0 = p$, $p_k = q$, $\vec{u}_0 = \vec{u}$, $\vec{u}_k =
    \vec{v}$;
  \item $p_{i-1}(\vec{u}_{i-1}) \xrightarrow{\pi_i}_{\C_i}
    p_i(\vec{u}_i)$, where $\C_i \in \{\B_1,\B_2\}$ for every $i \in [1,k]$;
  \item $\vec{u}_i \in H$ for each $i \in [1,k-1]$; and
  \item $k \leq |H| = (E+1)^2 \leq D^{O(1)}$.
  \end{enumerate} 

By combining~(ii) with Lemma~\ref{C Belt1} we have that for each run
$p_{i-1}(\vec{u}_{i-1}) \xrightarrow{\pi_i}_{\C_i} p_i(\vec{u}_i)$
there exists a linear path scheme $\rho_i =
\alpha_i(\beta_i)^*\gamma_i$ such that $p_{i-1}(\vec{u}_{i-1})
\xrightarrow{\rho_i}_{\C_i} p_i(\vec{u}_i)$ and $|\rho_i| \leq (|Q| +
\norm{T} + E)^{O(1)} = (|Q| + \norm{T} + D)^{O(1)}$. For simplicity,
here we only treat the case where each $\rho_i$ has precisely one
cycle, the cases when some $\rho_i$ contains no cycle can be dealt
with analogously. Note that whenever $i \in [2,k-1]$ we have
$p_{i-1}(\vec{u}_{i-1}), p_i(\vec{u}_i) \in Q \times H$ by (iii).
Since $\rho_i$ has only one cycle, $|\rho_i|\leq
(|Q|+\norm{T}+D)^{O(1)}$ and $\vec{u}_{i-1},\vec{u}_i\in H$ there
exists some $e_i \leq (|Q| + \norm{T} + E)^{O(1)}$ such that
$$p_{i-1}(\vec{u}_{i-1})
\xrightarrow{\alpha_i(\beta_i)^{e_i}\gamma_i}_{\C_i} p_i(\vec{u}_i),\quad\text{ thus in particular }\quad
p_{i-1}(\vec{u}_{i-1})
\xrightarrow{\alpha_i(\beta_i)^{e_i}\gamma_i}_{\L'} p_i(\vec{u}_i)\quad
$$
Consequently, we have
$$
p_0(\vec{u}_0) \xrightarrow{\alpha_1\beta_1^*\gamma_1}_{\C_1}
p_1(\vec{u}_1) \xrightarrow{\prod_{i=2}^{k-1}
  \alpha_i\beta_i^{e_i}\gamma_i}_{\L'} p_{k-1}(\vec{u}_{k-1})
\xrightarrow{\alpha_k\beta_k^*\gamma_k}_{\C_k} p_k(\vec{u}_k).
$$ 
Hence, we define 
$$\rho\quad \defeq\quad \alpha_1\beta_1^*\gamma_1 \cdot
  \left(\prod_{i=2}^{k-1} \alpha_i\beta_i^{e_i}\gamma_i\right) \cdot
  \alpha_k\beta_k^*\gamma_k$$ 
which has at most two cycles and for which we have
$p(\vec{u})\xrightarrow{\rho}_{\L'}q(\vec{v})$ and 
\begin{eqnarray*}
  |\rho| &\ \leq\ & k \cdot \max\{e_i : i \in [2,k-1]\} \cdot
  \max\{|\rho_i| : i \in [1,k]\} \\
  &\ \stackrel{\text{(iv)}}{\leq}\ & D^{O(1)} \cdot (|Q| + \norm {T} +
  E)^{O(1)} \cdot (|Q| + \norm{T} + D)^{O(1)} \\
  & = & (|Q| + \norm{T} + D)^{O(1)}.
\end{eqnarray*}
This concludes the proof of Proposition~\ref{P Belts}.

\subsection{Factorizing arbitrary runs: Proof of Theorem~\ref{main1}}\label{sec:glue}

By application of the results established in
Sections~\ref{sec:ultimately-flatness}
and~\ref{sec:bounded-component}, we will now prove
Theorem~\ref{main1}. In Section~\ref{sec:ultimately-flatness}, we
showed that the following two kinds of runs can be captured by small
linear path schemes:

\begin{itemize}
\item Type (1): Runs between two configurations $q(\vec{u})$ and
  $q(\vec{v})$ where both components of
  $\vec{u}$ and $\vec{v}$ are sufficiently large, but intermediate
  configurations could have small counter values.
\item Type (2): Runs on which for all configurations both counter
  values are sufficiently large.
\end{itemize}
Complementary, in Section~\ref{sec:bounded-component} we showed that
there are small linear path schemes with at most two cycles that
capture the following runs:
\begin{itemize}
\item Type (3): Runs on which for all configurations at least one counter value is not too large.
\end{itemize}

The goal of this section is to show that any run can be factorized
into few runs that are each of types~(1), (2) or (3). To this end, let
us fix a $2$-VASS $V=(Q,T)$. Let $D\leq (|Q| + \norm{T})^{O(1)}$ be
the constant from Proposition~\ref{P upper right}. Informally
speaking, we have hereby defined that ``sufficiently large'' means to
be greater or equal to $D$.  Moreover we set $\L \defeq
([0,D+\norm{T}] \times \N) \cup (\N \times [0,D+\norm{T}])$, $\O
\defeq [D,\infty)^2$, and $\B\defeq\L \cap
  \O=([D,D+\norm{T}]\times\N)\cup(\N\times[D,D+\norm{T}])$.  Again,
  informally speaking, we have hereby defined that ``not too large''
  means to be smaller or equal to $D+\norm{T}$.

Let us summarize what we have proven in
Sections~\ref{sec:ultimately-flatness}
and~\ref{sec:bounded-component}:
\begin{itemize}
\item Runs of type~(1) can be captured by a set of linear path schemes
  $R$, where each $\rho\in R$ has at most two cycles and length at
  most $(|Q|+\norm{T})^{O(1)}$ by Proposition \ref{P upper right}(a).
\item Runs of type~(2) can be captured by a set of linear path schemes
$X$, where each $\rho\in X$ has at most $2\cdot |Q|$ cycles and length
at most $(|Q|+\norm{T})^{O(1)}$ by Proposition \ref{P upper right}(b).
\item Runs of type~(3) can be captured by a set of linear path schemes
  $Y_\L$, where each $\rho\in Y_\L$ has at most two cycles and length
  at most $(|Q|+\norm{T}+D)^{O(1)}=(|Q|+\norm{T})^{O(1)}$ by
  Proposition \ref{P Belts}.
\end{itemize}
Given $p(\vec{u})$ and $q(\vec{v})$, let us fix an arbitrary run
$p(\vec{u}) \xrightarrow{\pi}_{\N^2} q(\vec{v})$, where $\pi=t_1\cdots
t_k\in T^k$ and
\begin{align*}
  p(\vec{u}) = q_0(\vec{u}_0) \xrightarrow{t_1}_{\N^2} q_1(\vec{u}_1)
  \cdots \xrightarrow{t_k}_{\N^2} q_k(\vec{u}_k) = q(\vec{v})\quad.
\end{align*}
 We will be interested in the
indices of configurations whose counter values lie in $\B$ and define
\begin{align*}
  I &\quad \defeq\quad \{ i\in [0,k] : \vec{u}_i \in \B \}\quad.
\end{align*}
Let us define the function $x:I\rightarrow I$ that maps each index
$i\in I$ to the smallest element in $I$ larger than $i$ (and $i$ if
$i=\max I$), i.e.\
$$
x(i)\quad\defeq\quad
\begin{cases}\min\{j\in I: j>i\} & \text{ if $i<\max I$,}\\
i &\text{ otherwise, i.e.\ $i=\max I$}\quad.
\end{cases}
$$ We also define the function $\ell:\{q_i\in Q: i\in I\}\rightarrow
I$ that maps each state $q$ that appears in a configuration in
$Q\times\B$ to the largest index in $I$ where it appears, i.e.\
$$
\ell(q)\quad\defeq\quad\max\{i\in I: q=q_i\}\quad.
$$ We are now interested in factorizing the run
$p(\vec{u})\xrightarrow{\pi}_{\N^2}q(\vec{v})$ into runs between
configurations that start and end in
$\B=\L\cap\O$.  More precisely, by the choice of $\O$, $\L$ and $\B$
and by the pigeonhole principle there exist indices $i_1,\ldots,i_h\in
I$ such that the run $p(\vec{u})\xrightarrow{\pi}_{\N^2}q(\vec{v})$
can be factorized as (cf. Figure~\ref{fig:thm1:decomp}):
\begin{multline*}
  q_0(\vec{u}_0)\xrightarrow{\pi_{0,1}}_{\D_{0,1}}
  q_{i_1}(\vec{u}_{i_1})\xrightarrow{\pi_1}_{\N^2}
  q_{\ell(q_{i_1})}(\vec{u}_{\ell(q_{i_1})})
  \xrightarrow{\pi_{1,2}}_{\D_{1,2}} q_{i_2}(\vec{u}_{i_2})
  \xrightarrow{\pi_2}_{\N^2}q_{\ell(q_{i_2})}(\vec{u}_{\ell(q_{i_2})})\quad\cdots\\
  \cdots\quad\xrightarrow{\pi_{h-1,h}}_{\D_{h,h-1}}q_{i_h}(\vec{u}_{i_h})
  \xrightarrow{\pi_h}_{\N^2}q_{\ell(q_{i_h})}(\vec{u}_{\ell(q_{i_h})})
  \xrightarrow{\pi_{h,h+1}}_{\D_{h,h+1}}
  q_k(\vec{u}_k)\quad,
\end{multline*}
where
\begin{enumerate}[(i)]
\item $h\leq |Q|$,
\item $i_t\in I$ and thus we have $\vec{u}_{i_t}\in\B$ and $q_{i_t}=q_{\ell(q_{i_t})}$ for each
  $t\in[1,h]$, 
\item $\D_{t,t+1}\in\{\O,\L\}$ for each $t\in[1,h]$, and
\item $i_{t+1}=x(\ell(q_{i_t}))$ for each $t\in[1,h-1]$.
\end{enumerate}

By (ii) each run of the form
$q_{i_t}(\vec{u}_{i_t})\xrightarrow{\pi_t}_{\N^2}q_{\ell(q_{i_t})}(\vec{u}_{\ell(q_{i_t})})$
is a run of type (1) and can hence be replaced by some linear path
scheme from $R$ (recall that $\B\subseteq\O$). By (iii) and (iv), each
run of the form $\xrightarrow{\pi_{t,t+1}}_{\D_{t,t+1}}$ is a run of
type~(2) or of type~(3) and can hence be replaced by some linear path
scheme from $X\cup Y_\L$. In summary, the run
$p(\vec{u})\xrightarrow{\pi}_{\N^2} q(\vec{v})$ can be replaced by a
linear path scheme that has at most $(h+1)\cdot2\cdot|Q|\leq O(|Q|^2)$
cycles and size at most
$(h+1)\cdot(|Q|+\norm{T})^{O(1)}=(|Q|+\norm{T})^{O(1)}$. This
concludes the proof of Theorem~\ref{main1}.

\begin{figure}[t]
  \begin{center}
    \includegraphics{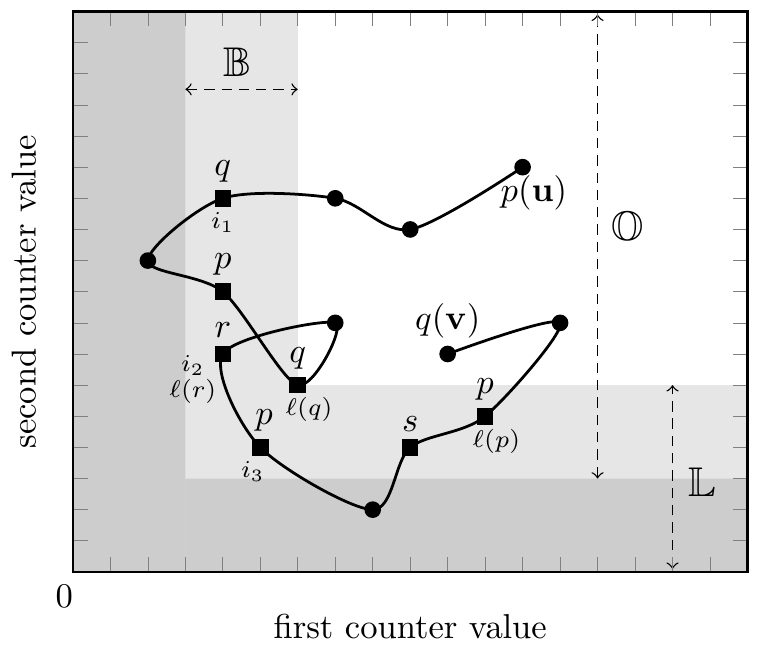}
  \end{center}
  \caption{Example of the decomposition of a path in the proof of
    Theorem~\ref{main1}. The region depicted is the positive quadrant
    in the Cartesian plane. Here, $I = \{3,5,6,8,9,11,12\}$ is marked
    with squares, and $i_1 = 3$, $\ell(q) = 6$, $i_2 = \ell(r) = 8$,
    $i_3 = 9$ and $\ell(p) = 12$.\label{fig:thm1:decomp}}
\end{figure}

\section{Complexity Results}{\label{geometry}}
Having established Theorem~\ref{main1}, it is now not difficult to
show that reachability in 2-VASS is in PSPACE by application of bounds
from integer linear programming. A complementary lower bound follows
via a reduction from reachability in bounded one-counter automata,
which is known to be PSPACE-complete~\cite{FJ13}.  This is the subject
of Section~\ref{sec:reachability-pspace} below which proves
Theorem~\ref{main2}. The PSPACE lower bound does, however, crucially
depend on binary encoding of numbers. In fact, we show in
Section~\ref{sec:reachability-unary} that reachability in unary 2-VASS
is in NP and NL-hard. The precise complexity of this problem remains
an open problem of this paper. Finally, for the sake of completeness,
in Section~\ref{sec:reachability-remarks} we briefly state some
corollaries of our results on the complexity of reachability in
$\Z$-VASS, and on coverability and boundedness in 2-VASS.

Before we begin, let us recall some definitions and results from
integer linear programming.  Let $A$ be a $d \times k$ integer matrix
and $\vec{c} \in \Z^d$. A \emph{system of linear Diophantine
  inequalities} (resp.\ a \emph{system of linear Diophantine
  equations}) is given as $\mathcal{I} : A\cdot\vec{x} \ge \vec{c}$
(resp.\ as $\mathcal{E}: A\vec{x} = \vec{c}$) and we say that
$\mathcal{I}$ (resp.\ $\mathcal{E}$) is {\em feasible} if there exists
some $\vec{e}\in \N^k$ such that $A\cdot \vec{e} \ge \vec{c}$
(resp.\ $A\cdot \vec{e} = \vec{c}$), i.e., every inequality
(resp.\ equality) holds in every row of $\mathcal{I}$
(resp.\ $\mathcal{E}$). Subsequently, we refer
to $\vec{e}$ as a {\em solution} of $\mathcal{I}$ or $\mathcal{E}$,
respectively. By $\eval{\mathcal{I}}\subseteq \N^k$ we denote the set
of all solutions of $\mathcal{I}$, the set of solutions $\eval{\mathcal{E}}\subseteq \N^k$
is defined analogously.

Let us now recall two bounds on solutions of systems of linear
Diophantine inequalities and equations that we subsequently rely
upon. The first bound we use in this paper concerns systems of linear
Diophantine inequalities.
\begin{proposition}[\cite{Schr86}, p.\ 239]
  \label{prop:schrijver-bound}
  Let $\mathcal{I} : A\cdot \vec{x} \ge \vec{c}$ be a feasible system
  of linear Diophantine inequalities, where $A$ is a $d \times k$
  matrix. Then there exists a solution $\vec{e} \in \N^k$ of
  $\mathcal{I}$ such that
  \begin{align*}
    \norm{\vec{e}}\quad \le\quad 2^{k^{O(1)}}\cdot O(\norm{A}+\norm{\vec{c}})\quad.
  \end{align*}  
\end{proposition}

Next, we consider a bound for feasible homogeneous systems of linear
Diophantine equations.
\begin{proposition}[\cite{Pot91}, Theorem~1]
  \label{prop:pottier-bound}
  Let $\mathcal{E} : A \cdot \vec{x} = \vec{0}$ be a system of linear
  Diophantine equations, where $A$ is a $d \times k$ integer matrix. Then
  there exists $P\subseteq \N^k$ such that $\norm{P}\le (\norm{A} +
  1)^d$ and
  \begin{align*}
    \eval{\mathcal{E}}\quad =\quad \cone_\N(P)\quad.
  \end{align*}
\end{proposition}
From this proposition it is now easy to generalize to the non-homogeneous
case.

\begin{corollary}
  \label{cor:pottier-bound}
  Let $\mathcal{E} : A \cdot \vec{x} = \vec{c}$ be a feasible system
  of linear Diophantine equations such that $A$ is a $d \times k$
  matrix. Then there exists a solution $\vec{e} \in \N^k$ of
  $\mathcal{E}$ such that
  \begin{align*}
    \norm{\vec{e}}\quad \leq\quad (\norm{A} + \norm{\vec{c}})^{O(d)}\quad.
  \end{align*}
\end{corollary}
\begin{proof}
  Define
  \begin{align*}
    \mathcal{E}'\quad:\quad \begin{bmatrix}A & -\vec{c}\end{bmatrix} 
    \begin{vmatrix} \vec{x} \\ y \end{vmatrix} \quad=\quad \vec{0}\quad,
  \end{align*}
  where $\vec{x}$ ranges over $\N^k$ and $y$ is a fresh variable
  ranging over $\N$. From Proposition~\ref{prop:pottier-bound} we have
  that $\eval{\mathcal{E}'}=\cone_\N(P)$ for some $P\subseteq
  \N^{k+1}$ such that $\norm{P}\le (\norm{A} + \norm{\vec{c}}+1)^{d}$.
  Now $\mathcal{E}$ is feasible if, and only if, there is some
  $\vec{p}\in\cone_\N(P)$ whose $(k+1)$-st component is equivalent to
  $1$. From such a $\vec{p}$ we obtain a solution of $\mathcal{E}$ with
  the desired bounds.
\end{proof}

\subsection{Reachability in 2-VASS is PSPACE-complete}
\label{sec:reachability-pspace}

In this section, we prove Theorem~\ref{main2} and show that
reachability in 2-VASS is PSPACE-complete. Given an instance
$p(\vec{u}) \xrightarrow{*}_{\N^2} q(\vec{v})$ of reachability, by
Theorem~\ref{main1} we have that $p(\vec{u}) \xrightarrow{\rho}_{\N^2}
q(\vec{v})$ for some linear path scheme $\rho$ such that $\abs{\rho}
\le (\abs{Q}+\norm{T})^{O(1)}$ and $\rho$ has $O(\abs{Q}^2)$
cycles. Writing $\rho = \alpha_0\beta_1^* \alpha_1 \cdots \beta_k^*
\alpha_k$, we have
\begin{align}
  \label{eqn:existential-reduction}
  p(\vec{u}) \xrightarrow{\rho}_{\N^2} q(\vec{v})\quad \iff\quad \text{there
    exist } e_1,\ldots,e_k\in\N \text{ such that } p(\vec{u})
  \xrightarrow{\alpha_0 \beta^{e_1}\alpha_1\cdots
    \beta^{e_k}\alpha_{k}}_{\N^2} q(\vec{v}).
\end{align}
Consequently, obtaining a PSPACE upper bound for reachability reduces
to bounding the binary representation of the $e_i$ polynomially in the
sizes of $V$, $\vec{u}$ and $\vec{v}$. Without loss of generality, in
the following we may assume that $e_i\ge 1$ for all $i\in[1,k]$.

Our approach is straightforward: we rephrase the existential question
from~(\ref{eqn:existential-reduction}) in terms of finding solutions
to a system of linear Diophantine inequalities and then apply standard
bounds from integer linear programming in order to bound the $e_i$.
For our reduction, let us first discuss the particular case when we
wish to decide whether the repetition of a cycle corresponds to a run.
In this case, it is sufficient to only check whether its initial and
final segments lead to counter values greater or equal to zero,
formalized by the following lemma.

\begin{lemma}
  \label{lem:ineq:cycles}
  Let $V=(Q,T)$ be $d$-VASS, $\vec{u} \in \N^d$ and let $\beta \in T^m$ be a
  cycle. Then there exists a system of linear Diophantine
  inequalities $\mathcal{I}: \vec{a} \cdot x \ge \vec{c}$ such that
  \begin{itemize}
  \item $e\in \eval{\mathcal{I}}$ if, and only if, $q(\vec{u}) \xrightarrow{\beta^e}_{\N^d} q(\vec{u} + e\cdot
    \delta(\beta))$ and $e\geq 1$ for every $e\in\N$,
  \item $\vec{a},\vec{c}\in \Z^{d+1}$, and
  \item $\norm{\vec{a}} \le \abs{\beta}\cdot \norm{T}$ and
    $\norm{\vec{c}}\le 2 \cdot \abs{\beta} \cdot  \norm{T} + \norm{\vec{u}}$.
  \end{itemize}
\end{lemma}
\begin{proof}
  Consider the following linear Diophantine inequalities containing
  two rows for every $1 \le j \le m$:
  \begin{align}
    \label{eqn:cycle-inequalities-1}\vec{u} + \delta(\beta[1,j]) & \quad\ge\quad \vec{0} \\
    \label{eqn:cycle-inequalities-2}\vec{u} + (x-1)\cdot \delta(\beta) + 
    \delta(\beta[1,j]) & \quad\ge\quad \vec{0}
  \end{align}
  The first row expresses that on the first traversal of $\beta$ we do
  not drop below zero. This row is independent from $x$, and if the
  constraints are infeasible we can chose $\mathcal{I}$ to be any
  infeasible system of linear Diophantine inequalities. 

  Next, in (\ref{eqn:cycle-inequalities-2}) we assert that the last
  time we traverse $\beta$ no counter drops below zero. In particular,
  we have
  \begin{align*}
    \vec{u} + (x-1)\cdot \delta(\beta) + \delta(\beta[1,j]) & \quad\ge\quad \vec{0}
    \qquad\iff\qquad \delta(\beta)\cdot x \quad\ge\quad \underbrace{\delta(\beta) -
      \delta(\beta[1,j]) - \vec{u}}_{\defeq \vec{c}_j}\quad.
  \end{align*}
  Consequently, we define $\vec{a}$ required in the lemma as
  $\vec{a}\defeq (1,\delta(\beta))$. For every $j$, let
  $\vec{c}_j=(c_{1,j},\ldots,c_{d,j})$, we set $\vec{c}$ to
  \begin{align*}
    \vec{c}\quad \defeq\quad (1, \max\{c_{1,1}, \ldots, c_{1,m}\}, \ldots,
    \max\{c_{d,1},\ldots, c_{d,m} \})\quad.
  \end{align*}
  The first row of $\mathcal{I}: \vec{a}\cdot x \ge \vec{c}$ asserts
  that any solution $e\in \eval{\mathcal{I}}$ is greater-equal to one,
  and the subsequent rows that \emph{all} constraints of type
  (\ref{eqn:cycle-inequalities-2}) are fulfilled by our particular
  choice of $\vec{c}$. In particular $q(\vec{u})
  \xrightarrow{\beta^e}_{\N^d} q(\vec{u} + e\cdot \delta(\beta))$. It
  is easily checked that the norms of $\vec{a}$ and $\vec{c}$ fulfill
  the requirements of the lemma.
\end{proof}

The restriction to non-zero solutions in Lemma \ref{lem:ineq:cycles}
is due the fact that the inequality constraints on prefixes of $\beta$
could wrongly exclude zero from a solution.  Therefore we have to
consider the cases when cycles are taken at least once or not at all
separately.  In doing so, we generalize the previous lemma to
arbitrary linear path schemes.  The function
$\text{sign}:\N\rightarrow\{0,1\}$ of naturals is defined as expected,
$\text{sign}(n)=1$ if $n\geq 1$ and $\text{sign}(n)=0$ if $n=0$.

\begin{lemma}
  \label{lem:lps-reachability-sldi}
  Let $V = (Q,T)$ be a $d$-VASS, $\vec{u} \in \N$ and $\rho = \alpha_0
  \beta_1^*\alpha_1 \cdots \beta_k^*\alpha_k$ be a linear path scheme
  from $p$ to $q$ and let $\chi:[1,k]\rightarrow\{0,1\}$. Then there
  exists a system of linear Diophantine inequalities
  $\mathcal{I}=\mathcal{I}(\vec{u},\rho,\chi)$ of the form
  $\mathcal{I} : A\cdot \vec{x} \ge \vec{c}$ such that
\begin{itemize}
\item 
$\vec{e}\in \eval{\mathcal{I}}$
if, and only if,
  $\pi=\alpha_0\beta_1^{e_1}\alpha_1 \cdots \beta^{e_k}\alpha_k$ and
  $p(\vec{u}) \xrightarrow{\pi}_{\N^d} q(\vec{u} + \delta(\pi))$ and
$\chi(i)=\text{sign}(e_i)$ for every $\vec{e}=(e_1,\ldots,e_k)\in \N^k$,
\item $A$ is a $((d+1)\cdot k) \times k$-matrix,
\item $\norm{A} \le k\cdot \abs{\rho}
  \cdot \norm{T}$, and 
\item $\norm{\vec{c}} \le O(\norm{\vec{u}} +
  \abs{\rho} \cdot \norm{T})$.
\end{itemize}
\end{lemma}
\begin{proof}
  We only prove the lemma for the concrete function $\chi:[1,k]\rightarrow\{0,1\}$, where $\chi(i)=1$ for all $i\in[1,k]$.
  In the following, we write $\vec{x} = (x_1,\ldots,x_k)$. First, we
  assert that the solutions $e_i$ are greater or equal to $1$, i.e.,
  \begin{align}
    \label{eqn:greater-zero}
    I_k \cdot \vec{x} & \ge \vec{1},
  \end{align}
  where $I_k$ is the $k$-th unit matrix and $\vec{1}=(1,\ldots,
  1)$. Next, informally speaking, we have to construct $\mathcal{I}$
  in a way such that we assert that the counter value does not drop
  below zero on any infix of $\rho$ in any dimension. For segments of
  $\rho$ between cycles, this can be ensured by the following
  constraints for every $j \in [0,k]$ and $\ell \in
  [1,\abs{\alpha_j}]$, which simply enforce the accumulated counter
  value to be non-negative:
  \begin{align}
    \notag & \vec{u} + \sum_{0\le i<j} \left( \delta(\alpha_i) + \delta(\beta_{i+1})
    \cdot x_{i+1} \right) + \delta(\alpha_j[1,\ell]) \quad\ge\quad \vec{0}\\
    \label{eqn:in-between}
    \iff\qquad & \sum_{1\le i \le j} \delta(\beta_{i})\cdot x_{i} \quad\ge\quad 
    -\vec{u} - \sum_{0\le i<j} \delta(\alpha_i) - \delta(\alpha_j[1,\ell])
  \end{align}

  For counter values which, informally speaking, occur along cycles
  $\beta_j$ of $\rho$, we follow the construction from
  Lemma~\ref{lem:ineq:cycles} and assert the following constraints for
  every $j \in [1,k]$ and $\ell \in [1,\abs{\beta_j}]$:
  \begin{align}
    \notag & \vec{u} +\delta(\alpha_0) +  \sum_{1 \le i< j} \left(\delta(\beta_{i})
    \notag \cdot x_i + \delta(\alpha_i) \right) + \delta(\beta_j[1,\ell])\quad \ge\quad
    \vec{0} \\
    \notag & \vec{u} + \delta(\alpha_0) +  \sum_{1 \le i< j} \left(\delta(\beta_{i}) 
    \cdot x_i + \delta(\alpha_i) \right) + \delta(\beta_j) \cdot (x_j - 1) +  
    \delta(\beta_j[1,\ell]) \quad\ge\quad \vec{0}\\
    \label{eqn:first-traversal}
    \iff\qquad & \sum_{1 \le i\le j-1} \delta(\beta_{i})
    \cdot x_i \quad\ge\quad -\vec{u} - \sum_{0\le i < j}\delta(\alpha_i) - 
    \delta(\beta_j[1,\ell])\\ 
    \label{eqn:last-traversal}
    & \sum_{1 \le i\le j} \delta(\beta_{i}) 
    \cdot x_i \quad\ge\quad -\vec{u} - \sum_{0 \le i < j} \delta(\alpha_i)
    + \delta(\beta_j) - \delta(\beta_j[1,\ell])
  \end{align}
  By our construction, it is easily verified that for every
  $\vec{e}=(e_1,\ldots,e_k)\N^k$ we have
$\chi(i)=1$ for all $i\in[1,k]$ and
  $p(\vec{u}) \xrightarrow{\alpha_0\beta_1^{e_1}\alpha_1 \cdots \beta^{e_k}\alpha_k}_{\N^d} q(\vec{u} + \delta(\pi))$ if,
  and only if, $\vec{e}$ fulfills \emph{all} constraints defined in
  (\ref{eqn:greater-zero}), (\ref{eqn:in-between}),
  (\ref{eqn:first-traversal}) and (\ref{eqn:last-traversal}). It thus
  remains to, informally speaking, extract the required system
  $\mathcal{I}$ of linear Diophantine inequalities from those
  constraints.

  For every fixed $j\in [1,k]$, by combining the constraints from
  (\ref{eqn:in-between}), (\ref{eqn:first-traversal}) and
  (\ref{eqn:last-traversal}), we obtain systems of linear Diophantine
  inequalities $\mathcal{I}_j': B_j\cdot \vec{x} \ge \vec{d}_j$ such
  that $B_j$ consists of at most $d$ \emph{different} rows, since
  every $x_i$ is multiplied by the \emph{same} $\delta(\beta_i)$. Let
  $A_j$ be the following $(d\times k)$-matrix: $A_j
  \defeq \begin{bmatrix} \delta(\beta_1) \cdots \delta(\beta_j) \;
    \vec{0} \cdots \vec{0}) \end{bmatrix}$.  For the $i$-th row of
  $A_i$, let $c_{j,i}\in \Z$ be the maximum value in $\vec{d}_j$ of
  the rows with the same coefficients in $\mathcal{I}_j'$, similar as
  in the construction of $\vec{c}_j$ in
  Lemma~\ref{lem:ineq:cycles}. We define $\vec{c}_j\defeq (c_{j,1},
  \ldots, c_{j,d})$ and set $\mathcal{I}_j: A_j\cdot \vec{x} \ge
  \vec{c}_j$. By construction, we now have that $\vec{e}\in \N^k$ is a
  solution of $\mathcal{I}_j$ if, and only if, $\vec{e}$ is a solution
  to $\mathcal{I}_j'$ and in particular fulfills all relevant
  constraints in (\ref{eqn:in-between}), (\ref{eqn:first-traversal})
  and (\ref{eqn:last-traversal}).

  In order to obtain the matrix $A$ and $\vec{c}$ required in the
  lemma, we define
  \begin{align*}
    A \quad\defeq\quad \begin{bmatrix} I_k \\ A_1 \\ \vdots \\ A_k \end{bmatrix}\qquad
    \text{ and }\qquad \vec{c} \quad\defeq\quad \begin{vmatrix} \vec{1}\\\vec{c}_1 \\
    \vdots\\ \vec{c}_k \end{vmatrix}\qquad.
  \end{align*}
  The dimension of $A$ and $\vec{c}$ is as required. It thus remains
  to estimate the norm of $A$ and $\vec{c}$. We have
  \begin{align*}
    \norm{A} \quad\le\quad \sum_{1\le i \le k} \norm{\delta(\beta_i)} \quad\le\quad k
    \cdot \abs{\rho} \cdot \norm{T}\quad.
  \end{align*}
  For $\vec{c}$, the following inequality bounds the norm of the
  right-hand sides of (\ref{eqn:in-between}),
  (\ref{eqn:first-traversal}) and (\ref{eqn:last-traversal}):
  \begin{align*}
    \norm{\vec{c}} & \quad\le\quad \norm{\vec{u}} + 2 \cdot \abs{\rho} \cdot \norm{T}
  \end{align*}
\end{proof}

By application of Proposition~\ref{prop:schrijver-bound}, this lemma
now enables us to give bounds on the length of a run witnessing
reachability for two given configurations.

\begin{lemma}\label{lem:runs:bounds}
  Let $V = (Q,T)$ be a $d$-VASS, let $p(\vec{u})$ and $q(\vec{v})$ be
  configurations of $V$, and let $\rho = \alpha_0 \beta_1^*\alpha_1 \cdots
  \beta_k^*\alpha_k$ be a linear path scheme from $p$ to $q$. Then $p(\vec{u})
  \xrightarrow{\rho}_{\N^d} q(\vec{v})$ if, and only if, $p(\vec{u})
  \xrightarrow{\pi}_{\N^d} q(\vec{v})$ for some $\pi = \alpha_0
  \beta_1^{e_1}\alpha_1 \cdots \beta_k^{e_k}\alpha_k$ such that $e_i
  \le 2^{k^{O(1)}}\cdot O(\norm{\vec{u}} + \norm{\vec{v}} + \abs{\rho}
  \cdot \norm{T})$ for each $i\in[1,k]$.
\end{lemma}
\begin{proof}
  The set of those $e_1,\ldots,e_k \in \N$ that achieve
  $\vec{u}+\delta(\pi)=\vec{v}$ can be obtained from the set of
  solutions of the system $\mathcal{E}:B\cdot \vec{x} = \vec{d}$ of
  linear Diophantine equations with unknowns $\vec{x} =
  (x_1,\ldots,x_k)$, where
  \begin{align*}
    A \quad\defeq\quad \begin{bmatrix}\delta(\beta_1) \cdots
      \delta(\beta_k) \end{bmatrix} \qquad\text{ and }\qquad \vec{d} \quad\defeq\quad \vec{v} -
    \vec{u} - \sum_{0 \le i \le k} \delta(\alpha_i)\qquad.
  \end{align*}
  The constraint matrix of $\mathcal{E}$ is of dimension $d\times k$
  and has norm bounded by $\abs{\rho}\cdot \norm{T}$. The norm of the
  right-hand side of $\mathcal{E}$ is bounded by $\norm{\vec{u}} +
  \norm{\vec{v}} + \abs{\rho} \cdot \norm{T}$.
  Let us fix an arbitrary $\chi:[1,k]\rightarrow\{0,1\}$.
  Lemma~\ref{lem:lps-reachability-sldi} yields a system of linear
  Diophantine inequalities $\mathcal{I}=\mathcal{I}(\vec{u},\rho,\chi)$
of the form $\mathcal{I} : A\cdot \vec{x} \ge \vec{c}$
  whose set of solutions $\vec{e}=(e_1,\ldots,e_k)\in\N$ corresponds to all runs $p(\vec{u})\xrightarrow{\pi}_{\N^2}q(\vec{u}+\delta(\pi))$, 
where $\pi=\alpha_0\beta_1^{e_1}\alpha_1\cdots\beta_k^{e_k}\alpha_k$
and $\chi(i)=\text{sign}(e_i)$ for all $i\in[1,k]$.
  Consequently, for any $(e_1,\ldots,e_k) \in
  \eval{\mathcal{I}} \cap \eval{\mathcal{E}}$ and $\pi = \alpha_0
  \beta_1^{e_1}\alpha_1 \cdots \beta_k^{e_k}\alpha_k$, we have
  $p(\vec{u}) \xrightarrow{\pi}_{\N^d} q(\vec{v})$ and $\chi(i)=\text{sign}(e_i)$ for all $i\in[1,k]$. 
 Now we obtain
  $\mathcal{I}\cap \mathcal{E}$ as
  \begin{align*}
    \mathcal{I}\cap \mathcal{E}\quad:\quad \begin{bmatrix} A\\ B\\ -B \end{bmatrix}
    \cdot \vec{x} \quad\ge\quad \begin{vmatrix} \vec{c}\\ \vec{d}\\ -\vec{d} \end{vmatrix}\quad.
  \end{align*}
  From Lemma~\ref{lem:lps-reachability-sldi} and our observations
  above we conclude that the norm of the constraint matrix of
  $\mathcal{I}\cap \mathcal{E}$ is bounded by $k\cdot \abs{\rho} \cdot
  \norm{T}$. Moreover, the norm on right-hand side is bounded by
  $O(\norm{\vec{u}} + \norm{\vec{v}} + \abs{\rho} \cdot \norm{T})$. By
  application of Proposition~\ref{prop:schrijver-bound}, the bounds on
  the solutions of $\mathcal{I}\cap\mathcal{E}$ follow.
\end{proof}

An immediate corollary of Lemma~\ref{lem:runs:bounds} and
Theorem~\ref{main1} is that reachability in 2-VASS is in PSPACE.
\begin{corollary}
  \label{cor:2-vass-reachability}
  Reachability in 2-VASS is in PSPACE.
\end{corollary}

\begin{proof}
  Let $V = (Q, T)$ be a 2-VASS and $p(\vec{u}), q(\vec{v})$ be
  configurations of $V$. By Theorem~\ref{main1}, there exists a set
  $S$ of linear path schemes such that
  \begin{itemize}
  \item $p(\vec{u}) \xrightarrow{*}_{\N^2} q(\vec{v})$ if, and only
    if, $p(\vec{u}) \xrightarrow{S}_{\N^2} q(\vec{v})$,
  \item $|\rho| \leq (|Q| + \norm{T})^{O(1)}$ for every $\rho \in S$,
    and
  \item each $\rho \in S$ has at most $O(|Q|^2)$ cycles.
  \end{itemize}

  By Lemma~\ref{lem:runs:bounds}, if $p(\vec{u})
  \xrightarrow{\rho}_{\N^2} q(\vec{v})$ for some $\rho=\alpha_0\beta_1^*\alpha_1\cdots\beta_k^*\alpha_k \in S$ then
  $p(\vec{u}) \xrightarrow{\pi}_{\N^2} q(\vec{v})$ for some $\pi=\alpha_0\beta_1^{e_1}\alpha_1\cdots\beta_k^{e_k}\alpha_k \in
  \rho$ such that $e_1,\ldots,e_k\in[0,e]$, where $e$ can be bounded as
  \begin{align*}
    e & \quad\leq\quad 2^{|Q|^{O(1)}}\cdot 
    O\left(\norm{\vec{u}} + \norm{\vec{v}} + (\abs{Q} + \norm{T})^{O(1)}\cdot \norm{T}
    \right)\\
    & \quad\leq\quad 2^{{(\abs{V} + \log \norm{\vec{u}} + \log \norm{\vec{v}})}^{O(1)}}\quad.
  \end{align*}
 Since $\abs{\pi}\le \abs{\rho}\cdot e$, the run
 $p(\vec{u})\xrightarrow{\pi}_{\N^2}q(\vec{v})$ can be guessed
 nondeterministically in polynomial space by storing only the
 intermediate configurations in an on-the-fly manner.  Consequently,
 reachability in 2-VASS in PSPACE.
\end{proof}

In order to complete the proof of Theorem~\ref{main2}, it remains to
show hardness for PSPACE. We reduce from reachability in bounded
one-counter automata, which is known to be
PSPACE-complete~\cite{FJ13}. A bounded one-counter automaton is given
by a tuple $V=(Q,T,b)$, where $(Q,T)$ is a 1-VASS and $b\in \N$ is a
bound encoded in binary. Let $\B=[0,b]$, given configurations
$p(u),q(u)$ of $V$ such that $u,v\in \B$, reachability is to decide
whether $p(u) \xrightarrow{*}_\B q(v)$.
\begin{lemma}
  \label{lem:2-vass-pspace-hard}
  Reachability in 2-VASS is PSPACE-hard.
\end{lemma}
\begin{proof}
  Let $V=(Q,T,b)$ be a bounded one-counter automaton, and let
  $V'\defeq (Q,T')$ be the 2-VASS obtained from $V$ by setting
  $T'\defeq \{ h(t) : t\in T\}$, where $h(p,z,q)\defeq
  (p,(z,-z),q)$. We define an injection $\varphi$ from 
  configurations of $V$ to configurations of $V'$ as follows:
  \begin{align*}
    \varphi(q(z)) & \quad\defeq\quad q(z,b-z)
  \end{align*}
  For any path $\pi$, it is now easily checked by induction on
  $\abs{\pi}$ that
  \begin{align*}
    p(u) \xrightarrow{\pi}_\B q(v) \text{ in $V$}\quad \iff\quad \varphi(p(u))
    \xrightarrow{h(\pi)}_{\N^2} \varphi(q(v)) \text{ in $V'$}\quad.
  \end{align*}
\end{proof}

This concludes the proof of Theorem~\ref{main2} and shows that
reachability in 2-VASS is PSPACE-complete.

\subsection{Reachability in 2-VASS with Unary Updates}
\label{sec:reachability-unary}
For unary 2-VASS we can show that reachability is in NP and NL-hard. 

Given a unary 2-VASS $V$, by Theorem~\ref{main1} whenever $p(\vec{u})
\xrightarrow{*}_{\N^2} q(\vec{v})$ then there exists a linear path
scheme $\rho = \alpha_0\beta_1^*\alpha_1 \cdots \beta_k^*\alpha_k$
whose length is \emph{polynomial} in $\abs{V}$ such that $p(\vec{u})
\xrightarrow{\rho}_{\N^2} q(\vec{v})$. Moreover, the proof of
Corollary~\ref{cor:2-vass-reachability} shows that there exist
$e_1,\ldots,e_k\leq 2^{(\abs{V} + \log \norm{\vec{u}} + \log
  \norm{\vec{v})})^{O(1)}}$ such that for
$\pi=\alpha_0\beta^{e_1}\alpha_1 \cdots \beta^{e_k}\alpha_k$, we have
$p(\vec{u}) \xrightarrow{\pi}_{\N^2} q(\vec{v})$. In particular, every
$e_i$ can be represented using a polynomial number of bits. Hence,
$(\rho,e_1,\ldots,e_k)$ may serve as a certificate that can be guessed
in polynomial time. It remains to show that this certificate can be
verified in polynomial time. Checking that $\rho$ is a linear path
scheme is easily verified in polynomial time. In order to check if
$p(\vec{u}) \xrightarrow{\pi}_{\N^2} q(\vec{v})$ in polynomial time we
can construct the system of linear Diophantine equations from
Lemma~\ref{lem:lps-reachability-sldi} and verify that
$\vec{e}=(e_1,\ldots,e_k)$ is a solution to this system. This shows
that reachability in unary 2-VASS is in NP.

NL-hardness of reachability trivially follows from NL-hardness of
reachability in directed graphs. Here, we wish to slightly strengthen
this result and remark that reachability is NL-hard already for unary
2-VASS, whose underlying graph corresponds structurally to a linear
path scheme (formally, every state lies on at most one cycle and the
deletion of all cycles yields a union of isolated vertices and a
cycle-free path, cf.~Figure~\ref{fig:lps} at the beginning of this
document). Let $G=(U,E)$ be a directed graph such that
$U=\{u_0,\ldots, u_{m-1}\}$ and $E=\{ e_0,\ldots,e_{n-1} \} \subseteq
U\times U$. We define an injection $h:U\to [0,m-1]^2$ as
$h(u_i)=(i,m-1-i)$ that relates vertices of $G$ with vectors from
bounded intervals.  Let $\ell\defeq m\cdot n-1$, the flat unary 2-VASS
$V=(Q,T)$ can now be defined as follows:
\begin{align*}
  Q & \quad\defeq\quad \{ q_0, q_0',\ldots, q_{\ell}, q_{\ell}' \}\\
  T & \quad\defeq\quad \phantom{\cup} \{ (q_j, \vec{0}, q_{j+1}) : j\in [0,\ell-1] \}\\
  & \quad\phantom{\defeq}\quad \cup \{ (q_j, -h(u_i), q_j'), (q_j', h(u_k), q_j) :
  e_j= (u_i,u_k), i=j\bmod n, i\in [0, \ell] \}\quad.
\end{align*}
Suppose we wish to decide whether $u_{m-1}$ is reachable from $u_0$, we claim that
this is the case if, and only if, $q_0(h(u_0)) \xrightarrow{*}_{\N^2}
q_{\ell}(h(u_{m-1}))$. Informally speaking, the vertex currently visited
along a path is encoded in the counter values of $V$. Every loop
between $q_j$ and $q_j'$ allows for simulating the edge $e_{(j \bmod
  n)}=(u_i,u_k)$ of $G$. The transition from $q_j$ to $q_j'$ can only
be traversed if the vertex encoded into the current counter values
corresponds to $u_i$. If we are able to reach $q_j'$, the transition
back to $q_j$ then updates the currently visited vertex to
$u_k$. Since a path from $u_0$ to $u_{m-1}$ of minimal length in $G$ traverses at most $m$ vertices, $\ell+1 =
m\cdot n$ states $q_j$ suffice.
\begin{theorem}
  Reachability in unary 2-VASS is in NP and NL-hard.
\end{theorem}

\subsection{Derived Results}
\label{sec:reachability-remarks}

In this section, we explicitly state and remark some results that can
additionally be derived from the technical results established in this
paper.

\subsubsection{$\Z$-Reachability in Unary $d$-VASS is NL-complete for each fixed $d$}

The decomposition estabished in Proposition~\ref{P zreach} enables us
to obtain a new result on $\Z$-reachability of $d$-VASS when $d$ is
fixed. The complexity of this problem depends on the encoding of
numbers as well as the dimension $d$. When numbers are encoded in
binary, reachability is NP-complete even when
$d=1$~\cite{HH14,HKOW09}, and reachability is also NP-complete when
numbers are encoded in unary and $d$ is part of the input to the
problem~\cite{HH14}. By application of Proposition~\ref{P zreach} and
Corollary~\ref{cor:pottier-bound}, we can solve the case of
reachability under unary encoding of numbers for each fixed dimension
$d$.

\begin{theorem}
  For every fixed $d\geq 1$, $\Z$-reachability in unary $d$-VASS is
  NL-complete.
\end{theorem}

\begin{proof}
  NL-hardness trivially follows from NL-hardness of reachability in
  directed graphs. Let $d\geq 1$ be fixed.
 Let $V = (Q, T)$ be a $d$-VASS and $p(\vec{u}),
  q(\vec{v}) \in Q \times \Z^d$ be two configurations as input to the $\Z$-reachability problem. 
By Proposition~\ref{P zreach}, there exists a finite set $S$ of linear
  path schemes such that
  \begin{itemize}
  \item $p(\vec{u}) \xrightarrow{*}_{\Z^d} q(\vec{v})$ if, and only
    if, $p(\vec{u}) \xrightarrow{S}_{\Z^d} q(\vec{v})$,
  \item $|\rho|\leq 2\cdot |Q|\cdot|T|$ for each $\rho \in S$, and
  \item each $\rho\in S$ has at most $|T|$ cycles.
  \end{itemize}

  Suppose $p(\vec{u}) \xrightarrow{*}_{\Z^d} q(\vec{v})$, then
  $p(\vec{u}) \xrightarrow{\rho}_{\Z^d} q(\vec{v})$ for some $\rho =
  \alpha_0 \beta_1^* \alpha_1 \cdots \beta_k^* \alpha_k \in S$ with
  $k\leq|T|$.  Let $\mathcal{E} : A \cdot \vec{x} = \vec{c}$ be the
  system of linear Diophantine equations, where
  $$A \defeq
  \begin{bmatrix}
    \delta(\beta_1) & \cdots & \delta(\beta_k)
  \end{bmatrix} \in \Z^{d \times k} \qquad \text{and} \qquad
  \vec{c} \defeq \vec{v} - \left(\vec{u} +
  \delta(\alpha_0\alpha_1\cdots\alpha_k)\right) \in \Z^d\qquad.
  $$ Then, we have
  \begin{eqnarray*}
    p(\vec{u}) \xrightarrow{\rho}_{\Z^d} q(\vec{v}) &\quad\iff\quad&
    p(\vec{u}) \xrightarrow{\alpha_0 \beta_1^{e_1} \alpha_1 \cdots
      \beta_k^{e_k} \alpha_k}_{\Z^d} q(\vec{v}) \text{ for some } \vec{e}=(e_1,
    \ldots, e_k) \in \N^k \\
    &\iff&  \vec{e} \in \eval{\mathcal{E}}\text{ for some }\vec{e}\in\N^{k}\quad.
  \end{eqnarray*}
  By Corollary~\ref{cor:pottier-bound}, if $\eval{\mathcal{E}}\neq
  \emptyset$ then $\mathcal{E}$ has a solution $\vec{e}$ such that
  $\norm{\vec{e}} \leq (\norm{A} + \norm{\vec{c}})^{O(d)}$. Hence, by
  definition of $A$ and $\vec{c}$, the norm of solutions can be
  bounded by some $b$, where 
  $$b \quad\leq\quad (|T| \cdot |\rho| \cdot \norm{T} + |\rho| \cdot \norm{T} +
  \norm{\vec{u}} + \norm{\vec{v}})^{O(d)} \quad\leq\quad ((|T| +
  \norm{T})^{O(1)} + \norm{\vec{u}} + \norm{\vec{v}})^{O(d)}\quad.$$ Since
  $\norm{T}, \norm{\vec{u}}$ and $\norm{\vec{v}}$ are encoded in unary
  (i.e. $|V|=|Q|+|T|\cdot d\cdot\norm{T}$)
  and $d$ is fixed, we obtain $b \leq |V|^{O(1)}$.

  Thus, $p(\vec{u}) \xrightarrow{\rho}_{\Z^d} q(\vec{v})$ implies that
  $p(\vec{u}) \xrightarrow{\pi}_{\Z^d} q(\vec{v})$ for some $\pi
  \in T^*$, where $|\pi| \leq b \cdot |\rho| \leq
  |V|^{O(1)}$. Therefore, in order to decide reachability it suffices
  to guess on-the-fly the intermediate configurations of a path of
  polynomial length from $p(\vec{u})$ to $q(\vec{v})$, which can be
  done nondeterministically in logarithmic space.
\end{proof}

\subsubsection{Boundedness and Coverability in $d$-VASS}
For the sake of completeness, here we wish to discuss some
consequences of PSPACE-hardness of reachability in 2-VASS to the
complexity of coverability and boundedness in $d$-VASS that were left
open in the literature.

The boundedness problem can be stated as follows.\\
\problemx{$d$-VASS Boundedness} {A $d$-VASS $V=(Q,T)$ and a 
  configuration $p(\vec{u})$.}  
         {Is $\{ q(\vec{v}) : p(\vec{u}) \xrightarrow{*}_{\N^d} q(\vec{v}) \}$ an infinite set?\\}
\medskip 

\noindent
The coverability problem can be stated as follows.\\

\problemx{$d$-VASS Coverability} {A $d$-VASS $V=(Q,T)$ and
  configurations $p(\vec{u})$ and $q(\vec{v})$.}  
           {Does there exist $\vec{w}\ge\vec{v}$ such that
           $p(\vec{u}) \xrightarrow{*}_{\N^d} q(\vec{w})$?\\}
\medskip 

\noindent
The complexity of boundedness and coverability for $d$-VASS in a fixed
dimension $d$ has been studied by Rosier \& Yen in~\cite{RY86}. They
show that both problems are PSPACE-complete for any fixed $d \ge
4$. Chan~\cite{Chan88} later noted that boundedness is already
PSPACE-complete for $d = 3$, leaving the case $d=2$ as an open
problem.

\begin{theorem}[\cite{RY86,Chan88}]
  Boundedness and coverability in $d$-VASS are PSPACE-complete for any
  fixed $d \ge 3$.
\end{theorem}

It is moreover known that for $d=1$ those problems are
NP-complete~\cite{Haa12}. From the results in~\cite{FJ13} and
Lemma~\ref{lem:2-vass-pspace-hard}, it now easy to improve the lower
bounds from~\cite{RY86,Chan88} and show that reachability and
coverability are PSPACE-complete for every fixed $d\ge 2$. An instance
of reachability between $p(u)$ and $q(v)$ in a bounded one-counter
automaton with bound $b$ can be reduced to boundedness and
coverability in 2-VASS by using the construction provided in
Lemma~\ref{lem:2-vass-pspace-hard} as a gadget and adding an extra
transition from $q$ to a fresh control state $r$. This transition
simply checks whether the current counter values are equal to
$(v,b-v)$ by subtracting this value from the counter, and $r$ has a
single self-loop which increments both counters by one, say. Together
with the upper bounds established in~\cite{RY86}, the above-mentioned
proof sketch yields the following theorem as a corollary.

\begin{corollary}
  Boundedness and coverability in $d$-VASS are PSPACE-complete for any
  fixed $d \ge 2$.
\end{corollary}

\section{Conclusion and Future Work} \label{sec:conclusion}

In this paper, we have located the complexity, i.e.,
PSPACE-completeness, of the reachability problem for $2$-VASS. We have
also noted that the coverability and boundedness problems for $2$-VASS
are PSPACE-complete. When numbers are encoded in unary we showed that
$\Z$-reachability in $d$-VASS is NL-complete for any fixed
$d$. Reachability for unary $2$-VASS was shown to be $\NL$-hard and in
$\NP$. Our approach does not immediately lead to a better upper bound
than $\NP$ mainly due to the following reason.
Our proof showed that the reachability relation can be captured by a
set of linear path schemes whose number of cycles is quadratically
bounded.  The matrix of the resulting system of linear Diophantine
inequalities thus has quadratically many columns and its smallest
solutions can thus become exponentially large.  The latter correspond
to the exponents of the cycles of the linear path scheme and hence of
the length of the run.

It could be interesting to study the reachability problem in $d$-VASS
with a single control state, known as $d$-VAS, for $2 \leq d \leq 5$,
since $5$-VAS (resp. $2$-VAS) are slightly more (resp. less) general
than $2$-VASS and have semi-linear reachability sets ~\cite{HP79}.  A
more challenging problem seems to be to obtain a first complexity
upper bound for reachability in 3-VASS.

\bibliographystyle{abbrv}
\bibliography{bibliography}

\end{document}